\documentclass{scrartcl}
\usepackage[english]{babel}
\usepackage{amsmath,amsthm,amssymb}
\usepackage{mathtools}
\usepackage{pifont}

\usepackage{dsfont}
\usepackage{enumitem}
\usepackage{physics}
\usepackage[utf8]{inputenc}
\usepackage[left=2.55cm,right=2.55cm,top=3.1cm,bottom=3.1cm]{geometry}
\usepackage{thmtools}
\usepackage{thm-restate}
\usepackage{comment}
\usepackage{bbm}
\usepackage{soul}
\usepackage{float}
\usepackage{longtable}
\usepackage{authblk}
\usepackage{chngcntr}
\usepackage{apptools}

\usepackage{tikz,bbold}
\usetikzlibrary{backgrounds,positioning,shapes.geometric,decorations.markings,arrows,knots,calc,decorations.pathmorphing,decorations.pathreplacing,calc,shapes,fit}

\allowdisplaybreaks

\newcommand{\rhoid}{\rho^{\textup{ideal}}}

\newcommand{\perr}{\textup{P}_{\neq}}
\newcommand{\Efin}{E_{\mathrm{fin}}}
\newcommand{\mybox}[1]{
    \draw[fill=black,opacity=0.4,draw=white,thick] (-1+#1/2,-1+#1/2) -- (1+#1/2,-1+#1/2) -- (1+#1/2,1+#1/2) -- (-1+#1/2,1+#1/2) -- cycle;
    \draw[fill=black,opacity=0.6,draw=white,thick] (-1-#1/2,-1-#1/2) -- (1-#1/2,-1-#1/2) -- (1+#1/2,-1+#1/2) -- (-1+#1/2,-1+#1/2) -- cycle;
    \draw[fill=black,opacity=0.4,draw=white,thick] (1-#1/2,-1-#1/2) -- (1+#1/2, -1+#1/2) -- (1+#1/2,1+#1/2) -- (1-#1/2,1-#1/2) -- cycle;
    \draw[fill=black,opacity=0.4,draw=white,thick] (-1-#1/2,-1-#1/2) -- (-1+#1/2, -1+#1/2) -- (-1+#1/2,1+#1/2) -- (-1-#1/2,1-#1/2) -- cycle;
    \draw[fill=black,opacity=0.2,draw=white,thick] (-1-#1/2,1-#1/2) -- (1-#1/2,1-#1/2) -- (1+#1/2,1+#1/2) -- (-1+#1/2,1+#1/2) -- cycle;
    \draw[fill=black,opacity=0.4,draw=white,thick] (-1-#1/2,-1-#1/2) -- (1-#1/2,-1-#1/2) -- (1-#1/2,1-#1/2) -- (-1-#1/2,1-#1/2) -- cycle;
}

\tikzset{snake it/.style={decorate,decoration=snake}}

\tikzstyle{limb}=[thick,line cap=round]

\newcommand{\Alice}{
    \def\H{2.2}
    \def\W{5.2}     
    \def\L{3.2}     
    \def\T{0.08}   
    \def\x{-0.1*\W} 
    \def\xA{0}
    \coordinate (H) at (\xA,\H);
    \draw[thick,line cap=round](H)++(-165:0.3) to[out=-140,in=60]++ (-130:0.3)to[out=65,in=-90,looseness=1.0]++ (80:0.45) to[out=90,in=120,looseness=1.4]++ (20:0.2); %
    \draw[thick] (H) circle (0.3);
    \draw[thick,line cap=round] (H)++(-140:0.3) to[out=80,in=-120,looseness=1.8]++ (40:0.6); %
    \draw[thick] (H)++(-90:0.3) coordinate (N) to[out=-95,in=95]++ (0,-0.40*\H) coordinate (P);
    \draw[thick,line cap=round] (N)++(-95:0.03) to[out=-65,in=177]++ (0.25*\H,-0.3*\H) coordinate (RH);
    \draw[thick,line cap=round] (N)++(-95:0.03) to[out=-120,in=90]++ (-0.08*\H,-0.4*\H);
    \draw[thick] (P) to[out=-70,in=95] (\xA+0.08*\H,0);
    \draw[thick] (P) to[out=-100,in=72] (\xA-0.06*\H,0);
}

\newcommand{\Bob}{
    \def\H{2.2}
    \def\W{5.2}     
    \def\L{3.2}     
    \def\T{0.08}    
    \def\x{-0.1*\W} 
    \def\xB{0}
    \draw[thick, xscale=-1] (\xB,\H) circle(0.3) coordinate (H);
    \draw[thick, xscale=-1] (H)++(-90:0.3) coordinate (N) to[out=-92,in=92]++ (0,-0.40*\H) coordinate (P);
    \draw[thick,line cap=round, xscale=-1] (N)++(-95:0.03) to[out=-105,in=90]++ (-0.02*\W,-0.4*\H);
    \draw[thick,line cap=round, xscale=-1] (N)++(-85:0.03) to[out=-80,in=90]++ (0.02*\W,-0.4*\H);
    \draw[thick,line cap=round, xscale=-1] (P) to[out=-92,in=82] (\xB-0.02*\W,\T);
    \draw[thick,line cap=round, xscale=-1] (P) to[out=-80,in=90] (\xB+0.02*\W,\T);
}

\newcommand{\Eve}{
    \def\H{3.0}    
    \def\h{2.2}   
    \def\w{0.36}    
    \def\r{0.3}   

    \draw[thick] (0.1*\h,\h) circle (\r) coordinate (H);
    \draw[thick] 
      (H)++(-110:\r) coordinate (N) --++ (-100:0.02*\h) coordinate (SH)
      to[out=-85,in=85]++ (0,-0.38*\h) coordinate (P);
    \draw[thick,line cap=round] (SH)++(-85:0.015*\h) to[out=-115,in=60,looseness=1.8]++ (-0.17*\h,-0.31*\h); 
    \draw[thick,line cap=round] (SH)++(-85:0.015*\h) to[out=-60,in=90]++ (0.5*\w,-0.4*\h); 
    \draw[thick] (P) to[out=-110,in=85] (-0.5*\w,0);
    \draw[thick] (P) to[out=-80,in=108] ( 0.5*\w,0);
    \fill[rotate=-20] 
      (H)++(170:\r) to[out=-90,in=-90,looseness=1.9]++ (0:0.6*\r) --++ (0:0.02*\r)
      to[out=-90,in=-90,looseness=1.6]++ (0:0.7*\r) --++ (90:0.04*\r) --++ (180:1.32*\r) -- cycle;
    \draw[thick,line cap=round](H)++(0:\r) to [out = 260, in= 110,looseness=1.0]++ (-80:1.2*\r) to [out = 200, in= 10,looseness=1.0]++ (180:2.5*\r) to [out = 50, in= 290,looseness=1.0]++ (70:0.275); 
    \draw[thick,line cap=round, xscale=-1] (H)++(180:0.3) to[out=10,in=180,looseness=0.5]++ (30:0.53);
    \draw[thick] (H)++(130:\r) --++ (70:0.8*\r) --++ (300:0.55*\r);
    \draw[thick] (H)++(70:\r) --++ (50:0.6*\r) --++ (280:\r);

}

\definecolor{LightGray}{RGB}{220,220,220}
\definecolor{myred}{RGB}{255, 19, 0}
\definecolor{myblue}{RGB}{14, 81, 167}
\definecolor{myorange}{RGB}{255, 129, 0}
\definecolor{mygreen}{RGB}{0, 146, 44}

\newcommand{\trdist}[2]{\left\| #1 - #2 \right\|_{1}}

\usepackage[colorlinks=true,linkcolor=myblue,citecolor=myblue,urlcolor=myblue]{hyperref}
\usepackage{doi}
\usepackage{cleveref}
\crefname{definition}{Definition}{Definitions}
\Crefname{definition}{Definition}{Definitions}
\crefname{theorem}{Theorem}{Theorems}
\Crefname{theorem}{Theorem}{Theorems}
\crefname{claim}{Claim}{Claims}
\Crefname{claim}{Claim}{Claims}
\crefname{lemma}{Lemma}{Lemmas}
\Crefname{lemma}{Lemma}{Lemmas}
\crefname{appendixlemma}{Lemma}{Lemmas}
\Crefname{lemma}{Lemma}{Lemmas}
\crefname{rem}{Remark}{Remarks}
\Crefname{rem}{Remark}{Remarks}
\crefname{prop}{Proposition}{Propositions}
\Crefname{prop}{Proposition}{Propositions}
\crefname{cor}{Corollary}{Corollaries}
\Crefname{cor}{Corollary}{Corollaries}
\crefname{section}{Section}{Sections}
\Crefname{section}{Section}{Sections}
\crefname{equation}{}{}
\Crefname{equation}{}{}
\crefname{figure}{Figure}{Figures}
\Crefname{figure}{Figure}{Figures}
\crefname{appendix}{Appendix}{Appendices}
\Crefname{appendix}{Appendix}{Appendices}
\crefname{table}{Table}{Tables}
\Crefname{table}{Table}{Tables}
\crefname{exmp}{Example}{Examples}
\Crefname{exmp}{Example}{Examples}
\crefname{footnote}{Footnote}{Footnote}
\Crefname{footnote}{Footnote}{Footnote}

\newtheorem{definition}{Definition}
\newtheorem{lemma}{Lemma}
\newtheorem{appendixlemma}{Lemma}[section]
\newtheorem*{remark}{Remark}

\newcommand\blfootnote[1]{%
  \begingroup
  \renewcommand\thefootnote{}\footnote{#1}%
  \addtocounter{footnote}{-1}%
  \endgroup
}

\title{\LARGE Defining Security in Quantum Key Distribution}

\author[1]{\large Carla Ferradini}
\author[1]{\large Martin Sandfuchs}
\author[2]{\large Ramona Wolf}
\author[1]{\large Renato Renner}

\affil[1]{Institute for Theoretical Physics, ETH Zurich, Zurich, Switzerland}
\affil[2]{Naturwissenschaftlich-Technische Fakultät, Universität Siegen, Siegen, Germany}
\date{\normalsize \today}


\date{\vspace{-5ex}}

\begin{document}

\maketitle

\begin{abstract} 
\blfootnote{Contacts: \href{mailto:cferradini@phys.ethz.ch}{cferradini@phys.ethz.ch}, \href{mailto:martisan@phys.ethz.ch}{martisan@phys.ethz.ch}, \href{mailto:ramona.wolf@uni-siegen.de}{ramona.wolf@uni-siegen.de}, and \href{mailto:renner@ethz.ch}{renner@ethz.ch}}
The security of quantum key distribution (QKD) is quantified by a parameter $\varepsilon>0$, which---under well-defined physical assumptions---can be bounded explicitly. This contrasts with computationally secure schemes, where security claims are only asymptotic (i.e., under standard complexity assumptions, one only knows that $\varepsilon \to 0$ as the key size grows, but has no explicit bound). Here we explain the definition and interpretation of $\varepsilon$-security. Adopting an axiomatic approach, we show that $\varepsilon$ can be understood as the maximum probability of a security failure. Finally, we review and address several criticisms of this definition that have appeared in the literature.
\end{abstract}

\section{Introduction}

Consider the following common cryptographic scenario: two parties, traditionally named Alice and Bob, wish to communicate over a channel that can be observed by an adversary, Eve (short for ``eavesdropper''). Alice’s objective is to transmit a secret message to Bob in such a way that Eve learns virtually nothing about its contents. This objective is achieved using encryption schemes.
Proving the security of such a scheme is highly non-trivial, since it requires accounting for every possible attack strategy Eve may run---of which there are infinitely many. Consequently, security cannot be demonstrated by experimental tests alone, not even in principle. This stands in stark contrast to the typical approach in physics, where experiments serve to validate theoretical models. Security, instead, must be established through theoretical analysis based on assumptions about the adversary. 

For most classical encryption schemes in use today, these assumptions concern the adversary's computing power. This is why the resulting security is termed \emph{computational}. For example, one assumes that the adversary cannot factor integers in polynomial time~\cite{RSA_1978}. However, such computational assumptions remain unproven and, in some cases, have turned out to be unjustified~\cite{Shor_1997}. Furthermore, even if they were valid, the resulting guarantees are only asymptotic. At best, one can state that ``the probability that an adversary learns something about the secret message (using a fixed amount of computational resources) decreases rapidly with increasing key length.'' What is missing are explicit quantitative bounds on that probability. Consequently, choosing key lengths to achieve a desired security level is largely a matter of educated guesswork (see Appendix~A of~\cite{RSA_1978} for examples of such estimates).

This issue is entirely avoided by \emph{information-theoretically secure} cryptosystems like those based on quantum key distribution (which this document is mainly about) and one-time-pad encryption. In contrast to schemes that rely on computational assumptions, these provide explicit, quantitative guarantees, such as ``the probability that an adversary learns something about the secret message (using arbitrary computational resources) is at most~$\varepsilon$.'' Here, $\varepsilon$ is a concrete security parameter, e.g., $\varepsilon = 10^{-10}$. Fig.~\ref{fig:quantcost} shows the key difference to computational security: if a message is encrypted with a classical protocol like RSA, its security degrades as the computational resources available to adversaries grow with time.

\begin{figure}[t]
	\centering
	\begin{tikzpicture}[scale=1.1]
		\draw[thick,->,>=stealth] (0,0) -- (9.25,0);
		\draw[thick,->,>=stealth] (0,0) -- (0,4);
		\node at (9.65,0.05) {\footnotesize time};
        \node at (9.65,-0.2) {\footnotesize (on scale of years)};
		\node at (0,4.7) {\footnotesize probability of};
		\node at (0,4.45) {\footnotesize security breach};
  		\node at (0,4.2) {\footnotesize (logarithmic scale)};      
		\draw[thick,color=mygreen] (0,0.5) -- (9,0.5);
		\draw[thick, color=myred] (0,0.25) -- (0,0.25) -- (1.5,0.5) -- (1.5,0.75) -- (3,1) -- (3,1.3) -- (4,1.5) -- (4,3.5) -- (9,3.5);
		\draw[thick, color=myblue] (0,0.25) -- (0,0.25) -- (1.25,0.3) -- (1.25,0.5) -- (2.5,0.75) -- (2.5,1.25) -- (3.75,1.5) -- (3.75,1.8) -- (5,2) -- (5,2.3) -- (6.5,2.5) -- (6.5,2.7) -- (7.25,2.9) -- (7.25,3.15) -- (9,3.45);
		\node[color=myblue] at (8.5,3) {\footnotesize post-quantum};
		\node[color=myblue] at (8.5,2.75) {\footnotesize protocol};
		\node[color=mygreen] at (8,0.75) {\footnotesize QKD\,+\,OTP protocol};
		\node[color=myred] at (8,4) {\footnotesize RSA};
		\node[color=myred] at (8,3.75) {\footnotesize protocol};
		\draw[color=black,dashed] (4,-0.15) -- (4,1.5);
		\node[color=black] at (4,-0.4) {\footnotesize first universal};
		\node[color=black] at (4,-0.65) {\footnotesize quantum computer};
		\node[color=black] at (0.8,1.5) {\footnotesize algorithmic};
		\node[color=black] at (0.8,1.25) {\footnotesize discovery};
		\draw[->,>=stealth] (1.4,1.2) to[bend left] (1.5,0.85);
		\draw[->,>=stealth] (1.55,1.4) to[bend left] (2.45,1.3);
		\node at (2,2.3) {\footnotesize evolution of};
		\node at (2,2.05) {\footnotesize hardware};
		\draw[->,>=stealth] (2.65,2) to[bend left] (3.25,1.55);
		\draw[->,>=stealth] (2.8,2.25) to[bend left] (4.25,2);
	\end{tikzpicture}
	\caption{\label{fig:quantcost}\textbf{Information-theoretic versus computational security.} The schematic plot shows the probability that a message encrypted today can be read by an adversary at a future time. With information-theoretically secure schemes, such as quantum key distribution (QKD)  combined with one-time-pad (OTP) encryption, this probability remains a fixed value~$\varepsilon$. By contrast, for computationally secure protocols such as RSA, the probability of a breach increases over time as computational resources advance, and may eventually reach~$1$, e.g., when efficient factoring becomes possible. \emph{(Reproduced from~\cite{RennerWolf2023} with permission from ETH Zurich.)}}
\end{figure}
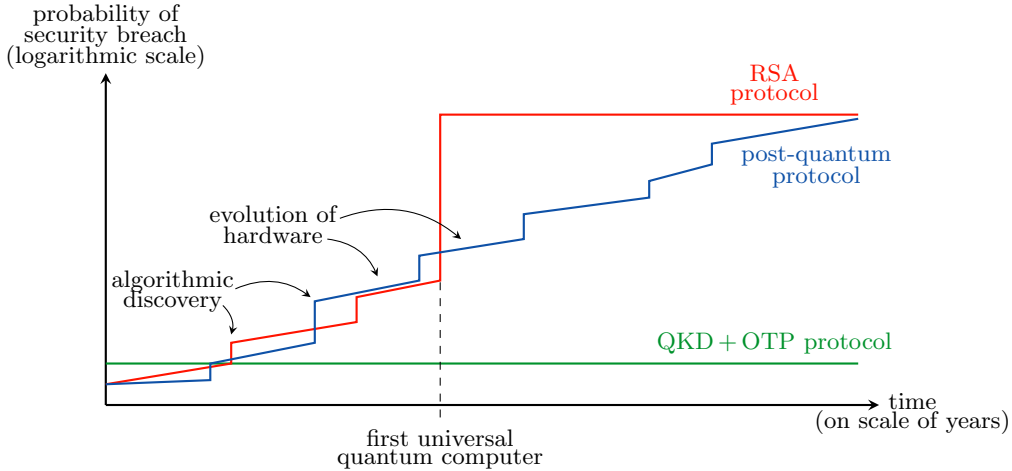

A well-known example of an information-theoretically secure encryption scheme is the \emph{one-time pad (OTP)} protocol~\cite{Shannon_1949}. While the OTP is resource-intensive---requiring a pre-shared secret key of length $n$ to securely transmit an $n$-bit message---it achieves perfect secrecy, i.e., $\varepsilon_{\mathrm{OTP}} = 0$. This means that the ciphertext reveals no information about the message to the adversary, provided that the pre-shared key was perfectly secret.
The OTP thus reduces the problem of securely transmitting an $n$-bit message from Alice to Bob to that of distributing a secret key of $n$ bits to them. This can be achieved through quantum key distribution (QKD)~\cite{BB_1984,Ekert_1991}: starting from a shared password (see Footnote~\ref{ftn_password}) and communicating over a potentially insecure quantum channel, Alice and Bob can generate a continuous stream of fresh secret key bits. Crucially, under suitable assumptions (detailed in \cref{sec_assumptions}), these key bits are information-theoretically secure, with a security parameter $\varepsilon_{\mathrm{QKD}} > 0$ that can be made arbitrarily small.

One might hope to prove that a QKD scheme generates a perfectly secure key, corresponding to $\varepsilon_{\mathrm{QKD}} = 0$. In practice, however, it is impossible to rule out the possibility that an attack succeeds with negligible (but non-zero) probability. For instance, imagine Eve measures each quantum signal sent from Alice to Bob and, by sheer luck, always chooses a basis that avoids introducing detectable disturbances. Alice and Bob would then proceed with the protocol and generate a key that is fully known to Eve. We must therefore accept a non-zero value of $\varepsilon_{\mathrm{QKD}}$. Nonetheless, any arbitrarily small $\varepsilon_{\mathrm{QKD}} > 0$ can be achieved by suitably choosing the protocol parameters. In typical QKD protocols, the required resources (such as the number of signals exchanged between Alice and Bob) scale only polylogarithmically with $1/\varepsilon_{\mathrm{QKD}}$.

The aim of this document is to formalise and explain the operational meaning of the security parameter~$\varepsilon$. Our guiding principle will be to interpret it as the \emph{(maximum) probability of a security failure}. For example, saying that a QKD protocol is $\varepsilon_{\mathrm{QKD}}$-secure means that---regardless of the attack strategy---no information about the key leaks to the adversary except with probability~$\varepsilon_{\mathrm{QKD}}$. 

Crucially, probabilities satisfy the union bound: the probability that a system with multiple components fails is no greater than the sum of the failure probabilities of its individual components. Interpreting $\varepsilon$ as a probability therefore naturally enables \emph{composability}: the security parameter of a composite system is upper-bounded by the sum of the security parameters of its parts.
A typical example is a secure message transmission scheme, constructed by composing QKD with OTP encryption, as described above. Since OTP encryption offers perfect secrecy ($\varepsilon_{\mathrm{OTP}} = 0$) and the QKD protocol is assumed to be $\varepsilon_{\mathrm{QKD}}$-secure, the overall security parameter of the system is simply $\varepsilon_{\mathrm{total}} = \varepsilon_{\mathrm{OTP}} + \varepsilon_{\mathrm{QKD}} = \varepsilon_{\mathrm{QKD}}$. This means that the adversary’s probability of gaining information about the message by attacking the scheme is bounded by this value.

This brings us to the central question regarding the security of a QKD scheme:
\begin{itemize}
    \item[\textbf{Q}] What does it mean that a key is secure except with probability~$\varepsilon_{\mathrm{QKD}}$?
\end{itemize}
We will address this question in detail in~\cref{sec:security-definition}. However, to formulate a meaningful answer, we must first establish a framework to model the knowledge an adversary might acquire about the key. This leads to the following preliminary question:
\begin{itemize}
\item[\textbf{Q1}] How to model the adversary’s knowledge?
\end{itemize}
Any security claim depends on specific assumptions---for instance, that quantum theory is correct and that Eve cannot access Alice’s and Bob’s laboratories. This raises the question:
\begin{itemize}
\item[\textbf{Q2}] What are the assumptions underlying security?
\end{itemize}
Answering these two preliminary questions will be the content of \cref{sec:assumptions}.

\section{Preliminaries and notation}
\Cref{tab:notation} summarizes some basic notions of quantum information theory as well as the notation used throughout the remainder of this document. For a detailed introduction to the formalism of quantum information theory, see, for instance, \cite{Nielsen_2010}. 

\vspace{-0.1cm}
\renewcommand{\arraystretch}{1.4}
\begin{longtable}[H]{|c||p{12cm}|}
        \hline
        \textbf{Notation} & \textbf{Description} \\
        \hline
        \hline
        $\mathcal{H}$ & Hilbert space\\
        \hline
        $\mathcal{H}_A, \mathcal{H}_B$, ... & Hilbert spaces belonging to different systems, $A$, $B$, ... \\
        \hline     
        $\mathcal{H}_X$ & Hilbert space belonging to a system $X$, equipped with a fixed orthonormal basis, $\{\ket{x}\}_{x \in \mathcal{X}}$, referred to as the computational basis \\
        \hline
        $S \geq 0$ & $S\in\mathrm{End}(\mathcal{H})$ is positive semi-definite \\
        \hline
        $S \leq T$ & $T - S \geq 0$, i.e., $T - S$ is positive semi-definite \\
        \hline
        $\mathcal{S}(\mathcal{H})$ & Set of normalised density operators (also referred to as quantum states of just states) on $\mathcal{H}$: $\mathcal{S}(\mathcal{H}) = \{{\rho\in \mathrm{End}(\mathcal{H})} : \tr(\rho) = 1, \rho \geq 0\}$ \\
        \hline
        $\rho_{AB}$ & Density operator acting jointly on systems $A$ and $B$: $\rho_{AB}\in\mathcal{S}(\mathcal{H}_A \otimes \mathcal{H}_B)$ \\
        \hline
        $\rho_A$ & Reduced density operator on $A$ obtained by tracing out $B$: $\rho_A = \tr_B[\rho_{AB}]$ \\

        \hline
        $\rho_X$ & Classical state on $\mathcal{H}_X$, describing a random variable $X$ with alphabet~$\mathcal{X}$: $\rho_X=\sum_{x \in \mathcal{X}} P_X(x) \ketbra{x}_X$, for a fixed computational basis $\{\ket{x}\}_{x}$ of $\mathcal{H}_X$  \\
        \hline
        $\rho_{XA}$ & Classical-quantum state describing a random variable $X$ correlated with a quantum system $A$: $\rho_{XA}=\sum_{x \in \mathcal{X}} P_X(x) \ketbra{x}_X \otimes \rho_{A|x}$, with $\rho_{A|x}$ the state of $A$ conditioned on $X=x$ \\
        \hline
        $\mathrm{Pr}[\Omega]$ & Probability of an event $\Omega\subseteq \mathcal{X}$: $\mathrm{Pr}[\Omega]=\sum_{x \in \Omega} P_X(x)$ \\
        \hline
        $\mathrm{Pr}_\rho[\Omega]$ & Same as $\mathrm{Pr}[\Omega]$ but with an explicit reference to the underlying distribution, i.e.,  $\mathrm{Pr}_\rho[\Omega] = \sum_{x\in\Omega} P_X(x)$ for some classical state $\rho_{X}$ as above \\
        \hline
        $\rho_{XA|\Omega}$ & State conditioned on $\Omega \subseteq \mathcal{X}$: $\rho_{XA|\Omega}=\frac{1}{\mathrm{Pr}[\Omega]} \sum_{x \in \Omega} P_X(x) \ketbra{x}_X \otimes \rho_{A|x}$ (Note that $\rho_{XA|\Omega}$ is only well-defined if $\mathrm{Pr}[\Omega] > 0$. This is unproblematic if $\rho_{XA|\Omega}$ is accompanied by a compensating factor of $\mathrm{Pr}[\Omega]$.) \\
        \hline
        $\rho_{XA \land \Omega}$ & Subnormalized conditional state given by $\rho_{XA \land \Omega} = \mathrm{Pr}[\Omega]\,\rho_{XA | \Omega}$ \\
        \hline
        $\mathcal{E}_{B|A}$ & Completely positive and trace preserving (CPTP) map or channel $\mathcal{E}_{B|A}\in\mathrm{Hom}\left(\mathrm{End}(\mathcal{H}_A),\mathrm{End}(\mathcal{H}_B)\right)$ \\
        \hline
        $\mathcal{I}_{R}$ & Identity channel on the system $R$ \\
        \hline
        $\mathcal{E}_{B|A}[\rho_{AR}]$ & Application of a channel to a state of a larger system: $(\mathcal{E}_{B|A} \otimes \mathcal{I}_R)[\rho_{AR}]$\\
        \hline
        POVM & A positive operator-valued measure, i.e., a set of positive semi-definite operators $\{M(x)\}_x$ such that $\sum_x M(x) = \mathds{1}$ \\
        \hline
        $\|S\|_1$ & Schatten $1$-norm of $S$, given by $\|S\|_1 = \tr[\sqrt{S^* S}]$ \\
        \hline
        $\frac{1}{2} \| \rho - \sigma \|_1$ & Trace distance between states $\rho$ and $\sigma$, defined by the Schatten $1$-norm \\
        \hline
    \caption{
    \textbf{Summary of notation.} Subscripts in capital letters refer to systems. For example, $\mathcal{H}_A$ denotes the Hilbert space associated with system $A$, and $\rho_A$ the quantum state assigned to it. We use $A, B, \ldots$ for generic quantum systems, while $X, Y, Z$ refer to classical systems, i.e., systems whose states are diagonal in a fixed computational basis.}
    \label{tab:notation}
\end{longtable}

\section{Setup and assumptions}
\label{sec:assumptions}

In this section, we introduce the conceptual framework within which security statements and proofs for QKD are phrased. 

\subsection{Modelling knowledge} \label{sec_knowledge}

We start by addressing the first preliminary question:
\begin{itemize}
    \item[\textbf{Q1}] How to model the adversary’s knowledge?
\end{itemize}
To answer this question, we proceed in two steps, corresponding to two subquestions. The first one is: 
\begin{itemize}
    \item[\textbf{Q1a}] How to model knowledge within quantum theory?
\end{itemize}
In (classical) information theory, knowledge is modelled using probabilities, i.e., we assign to every possible event a probability that represents how certain an agent is of the event happening.  This is also sometimes referred to as the Bayesian approach.
However, when describing knowledge about general quantum systems, probability distributions are not suitable\footnote{In principle, it is possible to represent knowledge about general quantum systems in terms of probability distributions instead of density operators, but this representation would not reflect the locality of subsystems~\cite{Bell_1964}.} and must be replaced by quantum states, which we will represent as density operators~\cite{Nielsen_2010}.\footnote{Using density operators to model knowledge is conceptually independent of any particular philosophical stance on quantum theory. In particular, the approach is fully compatible with realist interpretations---such as the many-worlds interpretation or Bohmian mechanics---which posit the existence of an objective ``quantum state of the universe.'' Within such frameworks, an agent’s knowledge of subsystems can still be captured by effective or conditional quantum states, which correspond to the density operators considered here.}  

Knowledge is inherently subjective, so it must always be specified relative to the agent holding it. Moreover, since an agent’s knowledge typically changes over time, it is equally important to specify the point in time at which the knowledge is defined.
Consider, for example, Alice and Bob describing a qubit. Initially, both have minimal information and assign the maximally mixed state to the qubit. Suppose Alice subsequently learns the outcome of a non-destructive measurement that was previously performed on the qubit. As a result, she updates her description and assigns a pure state, while Bob---still unaware of the measurement outcome---continues to describe the qubit using the maximally mixed state. Thus, Alice now associates a different state to the same system, reflecting her updated knowledge about it.

This subjective nature of knowledge raises another subquestion of \textbf{Q1}:
\begin{itemize}
    \item[\textbf{Q1b}] How should we deal with the subjective nature of knowledge?
\end{itemize}

It may be tempting to define security of a QKD protocol directly in terms of the density operator representing Eve’s knowledge about the generated key. However, this approach is problematic for at least two reasons. (1)~Density operators describe the knowledge a \emph{classical} agent has about a quantum system. Yet, an adversary may use quantum technologies to store and process information. Assuming Eve's knowledge is purely classical is therefore unjustified.\footnote{Indeed, restricting the adversary to classical knowledge leads to a significantly weaker notion of security than the standard security definition considered here~\cite{Konig_2007}.}
(2)~The goal of a security definition is to provide a guarantee to the honest parties---not the adversary---that the key they will produce by running a QKD protocol is secure. Having this guarantee thus corresponds to knowledge held by the honest parties.

Both concerns are addressed by considering a joint density operator of the form $\rho_{K E}$ that characterises Alice’s (or Bob’s) knowledge---prior to executing the QKD protocol---of the key~$K$ to be generated, along with all information~$E$ that may become accessible to Eve during the protocol. If, according to $\rho_{K E}$, Eve’s information~$E$ is uncorrelated with the key~$K$ they will generate, then Alice (or Bob) can be assured that~$K$ will be secure.

\begin{remark} \label{rem_stateterminology}
In quantum information theory, it is common to make statements such as ``Alice produces a state $\rho_S$''. Within the Bayesian approach adopted here, this phrasing can be misleading, as quantum states represent an agent’s knowledge rather than physical entities that a device can produce. Nonetheless, given the ubiquity of this terminology, we will also use it as a shorthand for: ``Alice prepares a system $S$ such that her knowledge of $S$ is described by~$\rho_S$.'' 
\end{remark}

\subsection{Assumptions} \label{sec_assumptions}

Claiming that a key generated by a QKD protocol is secure is ultimately a statement about the adversary’s (lack of) knowledge. However, this knowledge is not directly accessible through experiments. As a result, security cannot be established empirically. Instead, we must rely on security proofs, which ensure that the key is secure---provided certain assumptions hold. This leads to our second preliminary question:
\begin{itemize}
\item[\textbf{Q2}] What are the assumptions underlying security?
\end{itemize}
The assumptions refer to a typical QKD setup as illustrated in Fig.~\ref{fig_QKDsetup}. We begin by stating them informally, and then present a mathematical formulation.

\begin{enumerate}
    \item \emph{Quantum theory.} Alice, Bob, and Eve's actions are described correctly and completely by quantum theory.\footnote{\emph{Completeness of quantum theory} refers to the property that no agent can access any information other than what is modelled by quantum theory. Crucially, under the other assumptions that are usually made in quantum cryptography, notably the correctness of quantum theory and the existence of randomness within relativistic spacetime, completeness of quantum theory is implied~\cite{Colbeck_2011}; see also Section~A.1 of~\cite{Portmann_2022}.}
    \item \emph{Classical authenticated communication.}  Alice and Bob are connected via a two-way classical authenticated communication channel. That is, Eve cannot forge their messages without being detected.\footnote{\label{ftn_password}There exist protocols for achieving information-theoretically secure authentication using a short pre-shared key~\cite{Wegman_1981}. This key may be only weakly secret (such as a password generated from low-quality randomness~\cite{RennerWolf2003}) and contain errors~\cite{RennerWolf2004}. These authentication protocols can be composed with QKD protocols, as shown in~\cite{Portmann_2014_key_recycling}.} However, she may still eavesdrop on or block the communication.
    \item \emph{Quantum communication.} Alice and Bob are connected through a quantum channel, whose dynamics are described by a specified noise model when Eve acts only passively. (Nothing is assumed in the case where Eve is active. This assumption is thus only required to ensure that the QKD protocol completes successfully when Eve is passive, but it is not needed for the bare security claim; see \cref{sec:security-definition}.) 
    \item \emph{Trusted devices.} The local devices used by Alice and Bob reliably perform the tasks as specified by the protocol. In particular:
    \begin{enumerate}
        \item Random number generators (RNGs) produce random values according to the distribution specified by the protocol. These values are uncorrelated with any information held by Eve.
        \item Signal sources prepare the quantum states to be sent over the quantum communication channel according to the protocol prescription.
        \item Measurement devices apply the measurements to the signals received from the quantum communication channel according to the protocol prescription.
        \item Classical computing devices perform the computations specified by the protocol.
    \end{enumerate}
    \item \emph{Laboratory isolation.} Eve has no access to Alice's or Bob's laboratories. In particular, no information other than what is deliberately sent through the classical or quantum communication channels leaves the laboratories or becomes accessible to Eve.\footnote{This assumption is usually dropped in the \emph{real/ideal paradigm} (see, e.g., \cite{Portmann_2022}), where properties of the laboratory are not considered part of the security definition. According to this paradigm, one only requires that executing a QKD protocol (the ``real world'' scenario) be at least as secure as invoking a hypothetical machine (the ``ideal world'' scenario) that delivers perfect keys to Alice and Bob. Nothing is said about their laboratories; in particular, if they are not well isolated, the keys can leak both in the real and the ideal world.}
\end{enumerate}

\begin{figure}[t]
\centering
\begin{tikzpicture}[scale=0.92,
   mybaseline/.style={line width=1.15pt},
    mydashedline/.style={line width=1.pt, dashed},
]   
    \begin{scope}[scale=1.2,shift={(-4.5, 1)}]
        \begin{scope}[yscale=1*1.2, xscale=1.7*1.2]
            \mybox{0.5};
        \end{scope}
        \begin{scope}[shift={(-1.75, -1.2)},scale = 0.6,color=white]
            \Alice
        \end{scope}
        \begin{scope}[shift={(-1,-0.25)}]
            \begin{scope}[yscale=0.6*0.4, xscale=1*0.4, shift={(0,-0.2)}]
                \mybox{0.5};
            \end{scope}
            \node[color=white] at (-0.075,-0.075) {RNG};
        \end{scope}
        \begin{scope}[shift={(0,-0.45)}]
           \node at (0.6,-.4) {\includegraphics[width=0.1\textwidth]{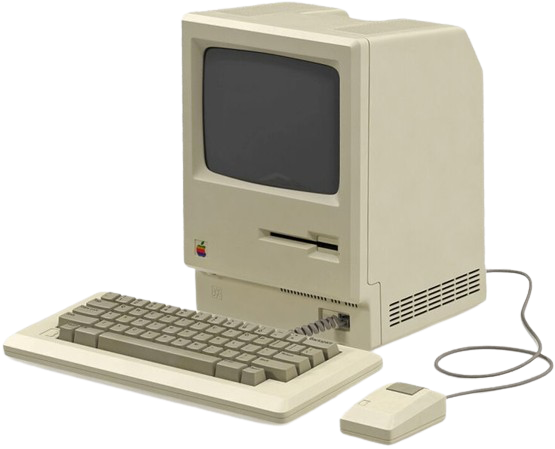}};
        \end{scope}
        \begin{scope}[shift={(0.65,0.4)}]
            \begin{scope}[yscale=0.6*0.4, xscale=1.6*0.4, shift={(0,-0.2)},rounded corners=0.1cm]
                \mybox{0.5};
            \end{scope}
            \node[color=white] at (-0.12,-0.1) {source};
        \end{scope}
        \draw[mybaseline, ->,color=white] (-1,-0.05) to [in=180,out=90] (-.2,0.4);
        \draw[mybaseline, ->,color=white] (-1,0.05-0.65) to [in=180,out=270] (-.2,-0.4-0.65);
        \begin{scope}[shift={(-1.5, -1.2)},transform shape,scale=0.45]
                    \draw[fill=myblue!25,draw=none,rounded corners=0.1cm] (-1.6, 3.7) -- (1.6, 3.7) -- (1.6, 4.4) -- (-1.6, 4.4) -- cycle;
                    \node[color=black] (Bits1) at (0, 4.05) {0 1 1 1 0 1 0 $\ldots$};
        \end{scope}
    \end{scope}
    \draw[thick,<->] (2.25, 0.5) -- (-2.25, 0.5);
    \node at (0,0.75) {\small cl.~authenticated channel};
    \draw[thick,snake it] (2.25, 1.65) -- (-2.25, 1.65);
    \node at (0,2.) {\small quantum channel};
    \begin{scope}[scale=1.2,shift={(4.5,1)}]
        \begin{scope}[yscale=1*1.2, xscale=1.7*1.2]
            \mybox{0.5};
        \end{scope} 
        \begin{scope}[shift={(0.2,-0.25)}]
            \begin{scope}[yscale=0.6*0.4, xscale=1*0.4, shift={(0,-0.2)}]
                \mybox{0.5};
            \end{scope}
            \node[color=white] at (-0.05,-0.05) {RNG};
        \end{scope}
        \begin{scope}[shift={(-1.2,0.4)}]
            \begin{scope}[yscale=0.6*0.4, xscale=1.6*0.4, shift={(0,-0.2)},rounded corners=0.1cm]
                \mybox{0.5};
            \end{scope}
            \node[color=white] at (-0.12,-0.1) {meas.};
        \end{scope}
        \begin{scope}[xscale = -1, shift ={(0.8,0)}]
            \draw[mybaseline, ->,color=white] (-1,-0.05) to [in=180,out=90] (-.2,0.4);
            \draw[mybaseline, ->,color=white] (-1,0.05-0.65) to [in=180,out=270] (-.2,-0.4-0.65);
        \end{scope} 
        \begin{scope}[shift={(0,-0.45)}]
           \node at (-1.2,-.4) {\scalebox{-1}[1]{\includegraphics[width=0.1\textwidth,]{Old_laptop.png}}};
        \end{scope}
        \begin{scope}[shift={(1.2, -1.2)},scale = 0.6,color=white]
            \Bob
        \end{scope}
        \begin{scope}[shift={(0.6, -1.2)},transform shape,scale=0.45]
            \draw[fill=myblue!25,draw=none,rounded corners=0.1cm] (-1.6, 3.7) -- (1.6, 3.7) -- (1.6, 4.4) -- (-1.6, 4.4) -- cycle;
            \node[color=black] (Bits1) at (0, 4.05) {0 1 1 1 0 1 0 $\ldots$};
        \end{scope}
    \end{scope}
        \begin{scope}[shift={(0, -2.5)},scale = 0.8,transform shape]
            \Eve
        \end{scope}
\end{tikzpicture}
\caption{\textbf{Setup considered in QKD.} Alice and Bob can communicate over a classical authenticated channel and an insecure quantum channel. The adversary Eve has access to all information outside Alice's and Bob's laboratories. Alice's and Bob's goal is to generate a key, i.e., a string of uniformly random bits known to them but not to Eve.}
\label{fig_QKDsetup}
\end{figure}
Some of these assumptions may seem overly strong. For example, it is practically impossible to produce random numbers that are perfectly uniformly distributed and uncorrelated to the information accessible to an adversary~\cite{FRT2013}. In fact, even the most carefully built RNG might slightly deviate from this ideal. Fortunately, such small deviations do not compromise the security of the resulting QKD scheme. Instead, the probability of a deviation can be accounted for by adding it to the overall security parameter. This is a consequence of \emph{composability}, which we will discuss in \cref{sec:security-definition}.

We now formalise the assumptions using the framework of quantum information theory. A QKD protocol can be understood as a sequence of actions to be carried out by Alice and Bob, as illustrated in Fig.~\ref{fig_QKD_protocol}. Similarly, an attack against the protocol corresponds to a sequence of actions by Eve. The following assumptions specify the mathematical representation of these actions.

\begin{enumerate}
  \item \emph{Quantum theory.} The information held by Alice, Bob, and Eve is described correctly and completely within the standard formalism of quantum information theory~\cite{Nielsen_2010}. In particular, all information is encoded in designated subsystems, each of which is associated with one of the agents (Alice, Bob, or Eve).
  Each action by an agent corresponds to a completely positive trace-preserving (CPTP) map acting on these systems.
  \item \emph{Classical authenticated communication.} Each use of the classical channel is described by a CPTP map
  $\smash{\mathcal{E}^{\mathrm{auth}}_{YE'|XE}}$ where $X$ is a classical input system associated with Alice (or Bob), and $Y$ is a classical output system associated with Bob (or Alice). The (possibly quantum) systems $E$ and $E'$ are associated with Eve. The map $\smash{\mathcal{E}^{\mathrm{auth}}_{YE'|XE}}$ may be chosen by Eve under the constraint that the value of $Y$ must either be equal to that of $X$, or take a special symbol indicating that the message was blocked. 
  \item \emph{Quantum communication.} Each use of the quantum channel is modelled by a CPTP map $\smash{\mathcal{E}^{\mathrm{quant}}_{BE'|AE}}$ where $A$ and $B$ are quantum systems associated to Alice and Bob, respectively (or by Bob and Alice, respectively, when communication occurs in the opposite direction) and where $E$ and $E'$ are associated to Eve. The map $\smash{\mathcal{E}^{\mathrm{quant}}_{BE'|AE}}$ can be chosen arbitrarily by Eve, but one particular choice corresponding to a fixed noise model is referred to as a \emph{passive} adversary.
  \item \emph{Trusted devices.} Each local action carried out by Alice or Bob is described by a CPTP map as specified by the protocol. 
  \item \emph{Laboratory isolation.} Each action executed by Eve is described by a CPTP map that acts as the identity on the systems associated to Alice and Bob.
\end{enumerate}
\begin{figure}[t]
\centering
\begin{tikzpicture}[
    green_box/.style={
        thick,
        text=black,
        draw=mygreen!90,
        rounded corners,
        minimum height=0.7cm,
        minimum width=0.9cm,
        fill=mygreen!60,
    },
    blue_box/.style={
        thick,
        text=black,
        draw=myblue!70,
        rounded corners,
        minimum height=0.7cm,
        minimum width=0.9cm,
        fill=myblue!50,
    },
]
    \node[] (Alice) at (-3, 0) {Alice};
    
    \node[] (Eve) at (0, 0) {Eve};

    \node[draw=myorange!70,fill=myorange!40,rounded corners=0.5cm,minimum width=7.5cm,minimum height=1cm] (rho) at (0, -0.8) {$\rho_{ABE}$};

    \node[] (Bob) at (3, 0) {Bob}; 
    \node[blue_box] (M) at ([yshift=-1.9cm] Alice.south) {$\mathcal{M}$};
    \draw[thick,->,>=stealth] (rho.south-|M.north) -- (M.north);
    \node[green_box,minimum width=1.4cm] (Eq) at ([yshift=-1.9cm] Eve.south) {$\mathcal{E}^{\mathrm{quant}}$};
    \draw[thick,->,>=stealth] (rho.south-|Eq.north) -- (Eq.north);
    \draw[thick,->,>=stealth,decorate,decoration={snake,amplitude=0.3mm,segment length=2.3mm}] (M.east) -- (Eq.west);

    \node[blue_box] (N) at ([yshift=-1.9cm] Bob.south) {$\mathcal{N}$};
    \draw[thick,->,>=stealth] (rho.south-|N.north) -- (N.north);
    \draw[thick,->,>=stealth,decorate,decoration={snake,amplitude=0.3mm,segment length=2.3mm}] (Eq.east) -- (N.west);

    \begin{scope}[on background layer]
        \draw[rounded corners,draw=lightgray!80!black,fill=lightgray!40]
            ([xshift=-0.25cm,yshift=0.25cm] M.north -| Alice.west) -- 
            ([xshift=0.25cm,yshift=0.25cm] N.north -| Bob.east) -- node[align=left,right,text=lightgray!40!black]{quantum\\ phase} 
            ([xshift=0.25cm,yshift=-0.3cm] N.south -| Bob.east) --
            ([xshift=-0.25cm,yshift=-0.3cm] M.south -| Alice.west) -- cycle;
    \end{scope}

    \node[green_box,minimum width=1.4cm,minimum height=0.7cm] (Ec1) at ([yshift=-1.4cm] Eq.south) {$\mathcal{E}^{\mathrm{auth}}$};
    \draw[thick,->,>=stealth] (Eq.south) -- (Ec1.north);
    \node[green_box,minimum width=1.4cm,minimum height=0.7cm] (Ec2) at ([yshift=-0.8cm] Ec1.south) {$\mathcal{E}^{\mathrm{auth}}$};
    \draw[thick,->,>=stealth] (Ec1.south) -- (Ec2.north);

    \node[blue_box] (PA) at (Ec1.center -| Alice.south) {$\mathcal{P}^{\mathrm{Alice}}$};
    \draw[thick,->,>=stealth] (M.south) -- (PA.north);

    \draw[thick,->,>=stealth] (PA.east) -- (Ec1.west);
    \draw[thick,->,>=stealth] (Ec1.east) -- (Ec1.east -| Bob.south);

    \node[blue_box] (PB) at (Ec2.center -| Bob.south) {$\mathcal{P}^{\mathrm{Bob}}$};
    \draw[thick,->,>=stealth] (N.south) -- (PB.north);
    \draw[thick,->,>=stealth] (PB.west) -- (Ec2.east);
    \draw[thick,->,>=stealth] (Ec2.west) -- (Ec2.west -| Alice.south);

    \begin{scope}[on background layer]
        \draw[rounded corners,draw=lightgray!80!black,fill=lightgray!40]
            ([xshift=-0.25cm,yshift=0.25cm] PA.north -| Alice.west) -- 
            ([xshift=0.25cm,yshift=0.25cm] PA.north -| Bob.east) -- node[align=left,right,text=lightgray!40!black]{classical\\ communication}
            ([xshift=0.25cm,yshift=-0.25cm] PB.south -| Bob.east) --
            ([xshift=-0.25cm,yshift=-0.25cm] PB.south -| Alice.west) -- cycle;
    \end{scope}

    \node[blue_box,minimum width=1.1cm] (KA) at ([yshift=-1.7cm] Ec2 -| Alice.south) {$\mathcal{K}^{\mathrm{Alice}}$};
    \draw[thick,->,>=stealth] (PA.south) -- (PA.south |- KA.north);
    \draw[thick,->,>=stealth] (KA.south) -- ([yshift=-0.7cm] KA.south) node[below]{$K_{A}$};

    \node[blue_box,minimum width=1.1cm] (KB) at ([yshift=-1.7cm] Ec2 -| Bob.south) {$\mathcal{K}^{\mathrm{Bob}}$};
    \draw[thick,->,>=stealth] (PB.south) -- (PB.center |- KB.north);
    \draw[thick,->,>=stealth] (KB.south) -- ([yshift=-0.7cm] KB.south) node[below]{$K_{B}$};

    \draw[thick,->,>=stealth] (Ec2.south) -- ([yshift=-0.7cm] Ec2.south |- KA.south) node[below]{$\Efin$};

    \begin{scope}[on background layer]
        \draw[rounded corners,draw=lightgray!80!black,fill=lightgray!40]
            ([xshift=-0.25cm,yshift=0.25cm] KA.north -| Alice.west) -- 
            ([xshift=0.25cm,yshift=0.25cm] KB.north -| Bob.east) -- node[align=left,right,text=lightgray!40!black]{classical\\ post-processing} 
            ([xshift=0.25cm,yshift=-0.25cm] KB.south -| Bob.east) --
            ([xshift=-0.25cm,yshift=-0.25cm] KA.south -| Alice.west) -- cycle;
    \end{scope}

    \draw[transparent] (-7, 0) -- (7, 0);
\end{tikzpicture}
\caption{\textbf{Quantum information-theoretic model of a QKD protocol.} The diagram shows the subsystems associated with Alice, Bob, and Eve, as well as the operations they perform.
The maps $\mathcal{M}$ and $\mathcal{N}$ represent Alice’s uses of the signal source and Bob’s uses of his measurement device, respectively. The maps $\mathcal{P}^{\text{Alice}}$, $\mathcal{P}^{\text{Bob}}$, $\mathcal{K}^{\text{Alice}}$, and $\mathcal{K}^{\text{Bob}}$ correspond to classical processing steps. This example is deliberately simplified; in general, QKD protocols can involve more complex structures, such as interleaved classical and quantum communication or multiple rounds of classical interaction.}
\label{fig_QKD_protocol}
\end{figure}
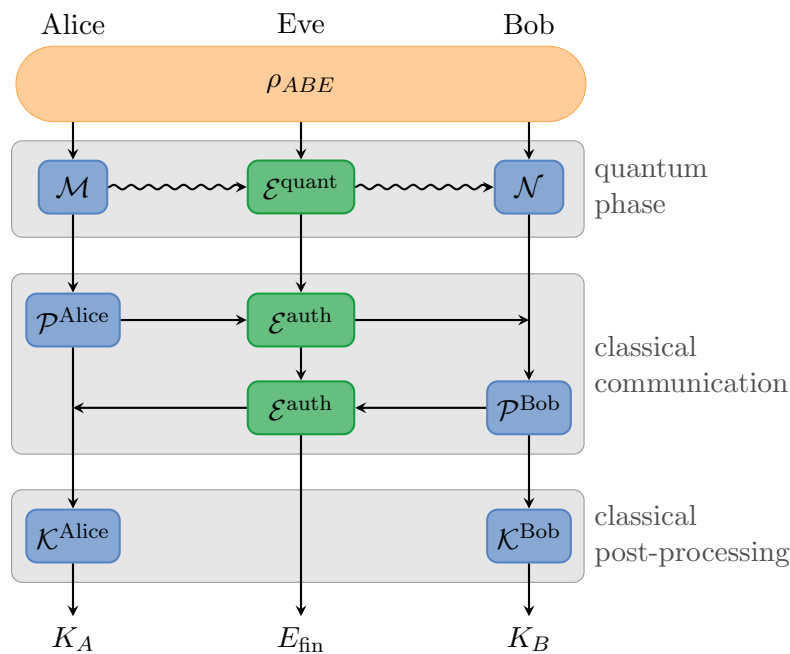

\section{Quantifying security}
\label{sec:security-definition}

We now turn to the main question posed in the introduction.\footnote{As we will focus on QKD, we omit the subscript and write $\varepsilon$ instead of $\varepsilon_{\mathrm{QKD}}$.}
\begin{itemize}
    \item[\textbf{Q}] What does it mean that a key is secure except with probability~$\varepsilon$?
\end{itemize}

To answer this question, it is worth taking inspiration from other engineering disciplines. Consider the design of a safety-critical structure, such as a bridge. Ideally, we would like to guarantee that the bridge will never collapse. In practice, however, it is impossible to account for every possible future scenario---rare combinations of adverse conditions might still lead to failure. No bridge can have a zero probability of collapse. Yet, we continue to build and use them. That is, we accept a bridge as safe if the probability of such a failure, which we may take as the bridge's security parameter~$\varepsilon$, is very small.

We seek a similar guarantee for a QKD scheme: except with negligible probability $\varepsilon$, the key produced by the scheme is secure, i.e., uniform and independent of any side information an adversary may hold. 
To formulate such a claim, we need a precise notion of what it means for two quantum states to be equal with high probability. In classical cryptography, where one considers probability distributions rather than quantum states, the \emph{Coupling Lemma} provides a satisfactory answer to this question \cite[Appendix A.3]{Portmann_2022}. In the quantum setting, however, the situation is more subtle, as the Coupling Lemma does not directly apply. To address this, we adopt an axiomatic approach in which the axioms are motivated operationally.

\subsection{Distinctness probability}
\label{sec:trace-distance}

We introduce a function $\perr: \mathcal{S}(\mathcal{H}) \cross \mathcal{S}(\mathcal{H}) \to \mathbb{R}_{\geq 0}$ that measures the probability that two situations, described by two different quantum states $\rho$ and $\sigma$ from $\mathcal{S}(\mathcal{H})$, are distinct.  Rather than starting with an explicit definition, we begin with a set of axioms that we demand the function to satisfy.

\begin{enumerate}[label=\textbf{A\arabic*}]
    \item \label{itm:triangle_ineq} \textbf{Triangle inequality:} For all $\omega\in\mathcal{S}(\mathcal{H})$
    \begin{equation} \label{eq_trianglebound}
        \perr(\rho,\sigma) \leq \perr(\rho,\omega) +\perr(\omega,\sigma).
    \end{equation}
    \item \label{itm:symmetry} \textbf{Symmetry:}  $\perr(\rho,\sigma)=\perr(\sigma,\rho)$.
    \item \label{itm:operational_bound} \textbf{Operational bound:} If $\rho = (1-\varepsilon) \sigma + \varepsilon \sigma'$ for some $\varepsilon \in [0,1]$ and $\sigma' \in \mathcal{S}(\mathcal{H})$ then\footnote{This bound need not be tight; in fact, it is trivial if $\rho$ is pure, corresponding to an extremal point of $\mathcal{S}(\mathcal{H})$.}
    \begin{equation} \label{eq_errorbound}
        \perr(\rho, \sigma) \leq \varepsilon.
    \end{equation}
    Note that, in particular, $\perr(\rho, \sigma) \leq 1$ holds for all $\rho$ and $\sigma$. Furthermore, we have $\perr(\rho,\sigma) = 0$ if $\rho = \sigma$.
\end{enumerate}

\begin{remark} 
  \sloppy One possible choice for $\perr$ is (a rescaled version of) the trace distance, i.e.,  $\perr(\rho, \sigma) = \frac{1}{M} \| \rho - \sigma\|_1$, for $M \geq 2$. It is straightforward to check that it satisfies the axioms.
\end{remark}

To motivate the axioms, imagine the ``event'' $\rho \neq \sigma$, in which the situations described by $\rho$ and $\sigma$ are distinct. (We emphasise that $\rho \neq \sigma$ is not an event in the usual sense of probability theory,\footnote{This is because $\rho$ and $\sigma$ are not random variables.} and we use this notion merely to provide intuition within this paragraph.)  Axiom~\ref{itm:triangle_ineq} may then be seen as the statement that if $\rho \neq \sigma$ then either $\rho \neq \omega$ or $\omega \neq \sigma$ (or both). Inequality~\eqref{eq_trianglebound} thus corresponds to the union bound from probability theory. Axiom~\ref{itm:symmetry} reflects the idea that $\rho \neq \sigma$ should be equivalent to $\sigma \neq \rho$. Finally, Axiom~\ref{itm:operational_bound} is motivated as follows: If $\rho = (1 - \varepsilon) \sigma + \varepsilon \sigma'$, then we can think of $\rho$ as being produced by picking $\sigma$ with probability $1 - \varepsilon$ and $\sigma'$ with probability $\varepsilon$. Hence, with probability $1 - \varepsilon$, we have picked $\rho=\sigma$, which means that the probability by which they differ can be at most $\varepsilon$.

In cryptographic security, it is prudent to adopt a conservative approach when quantifying the probability of failure. This motivates the following definition.

\begin{definition}[Maximum distinctness probability] \label{def:perr_max}
    The maximum distinctness probability $\perr^{\max}: \mathcal{S}(\mathcal{H}) \times \mathcal{S}(\mathcal{H}) \mapsto [0, 1]$ is the largest function  satisfying Axioms~\ref{itm:triangle_ineq}--\ref{itm:operational_bound}.
\end{definition}

It is not obvious that $\perr^{\max}$ is well-defined by this requirement, i.e., that such a maximum function actually exists. However, the following lemma provides an explicit construction in terms of an optimisation problem. Note also that $\perr^{\max}$, being defined as a maximum over functions, is unique.

\begin{restatable}[]{lemma}{perrmaxexistence} 
\label{lem:perr_max_existence}
  The maximum distinctness probability is given by
    \begin{equation} \label{eq:perrconstruction}
        \perr^{\max}(\rho, \sigma) = \inf \sum_{i=0}^{n-1} \hat{\delta}(\omega_{i}, \omega_{i+1}),
    \end{equation}
    where the infimum runs over all sets of states $\{\omega_{i}\}_{i=0}^{n}$ such that $\omega_{0} = \rho$ and $\omega_{n} = \sigma$. Here $\hat{\delta}(\rho, \sigma) \coloneqq \min\{ \tilde{\delta}(\rho, \sigma), \tilde{\delta}(\sigma, \rho) \}$ with
    \begin{equation} 
    \begin{aligned}
        \tilde{\delta}(\rho,\sigma)  \coloneqq& \inf_{\substack{\varepsilon\in[0,1]\\\sigma'\in \mathcal{S}(\mathcal{H})}} \varepsilon \\
        &\text{s.t. } \rho = (1 - \varepsilon) \sigma + \varepsilon \sigma'.
    \end{aligned}
    \end{equation}
\end{restatable}

The proof of this lemma, as well as the claims below, can be found in \cref{sec:proofs}. 

We continue with the description of other properties that we expect such a measure of distinctness to satisfy.

\begin{enumerate}[label=\textbf{P\arabic*}]
    \item \label{itm:positive_definite} \textbf{Positive definiteness:} $\perr(\rho,\sigma)=0$ if and only if $\rho=\sigma$.
    \item \label{itm:data_processing} \textbf{Data processing:} For any CPTP map $\mathcal{E}_{B|A}$,
    \begin{equation}
        \perr\left(\mathcal{E}_{B|A}[\rho_{A}], \mathcal{E}_{B|A}[\sigma_{A}]\right) \leq \perr(\rho_{A}, \sigma_{A}).
    \end{equation}
\end{enumerate}

Recall that the ``if'' direction of Property~\ref{itm:positive_definite} is immediate from~\ref{itm:operational_bound}. The converse, however, does not follow from the axioms: for example, if $\perr$ were defined as the constant zero function, then \ref{itm:triangle_ineq}–\ref{itm:operational_bound} would be satisfied, but~\ref{itm:positive_definite} would fail. Nonetheless, the requirement is reasonable if we want to think of the probability $\perr(\rho, \sigma)$ to be a faithful indicator for the differences between $\rho$ and $\sigma$. Property~\ref{itm:data_processing} is also a natural requirement because $\rho_A = \sigma_A$ obviously implies that $\mathcal{E}_{B|A}[\rho_A] = \mathcal{E}_{B|A}[\sigma_A]$. 

We will also be interested in the relation between the maximum distinctness probability and known distance measures, notably the trace distance, which is defined by the Schatten $1$-norm. The following upper bound holds for any $\perr$ satisfying \ref{itm:triangle_ineq}–\ref{itm:operational_bound}.

\begin{restatable}[]{lemma}{perrupperbound}
    \label{lem:perr_upper}
    Let $\perr: \mathcal{S}(\mathcal{H}) \cross \mathcal{S}(\mathcal{H}) \to \mathbb{R}_{\geq 0}$ satisfy Axioms~\ref{itm:triangle_ineq}--\ref{itm:operational_bound}. Then
    \begin{equation} \label{eq_perr_upper}
        \perr(\rho, \sigma) \leq \| \rho - \sigma \|_{1}
    \end{equation}
    holds for all $\rho, \sigma \in \mathcal{S}(\mathcal{H})$.
\end{restatable}

\begin{remark}
    Alternatively, the bound~\cref{eq_perr_upper} could be derived from Axiom~\ref{itm:triangle_ineq} and a symmetrized version of \cref{eq_errorbound} in Axiom~\ref{itm:operational_bound}, i.e., symmetry for all states (Axiom~\ref{itm:symmetry}) is not necessary.
\end{remark}

Because $\perr(\rho, \sigma) = \frac{1}{M} \| \rho - \sigma\|_1$, for $M \geq 2$, is a possible choice (see the earlier remark), the converse to \cref{lem:perr_upper} cannot be true in general. However, the following lemma shows that the maximum distincness probability $\perr^{\max}$ satisfies such a bound and that it also has the other properties discussed above.

\begin{restatable}[]{lemma}{perrlowerbound} \label{lem:perr_lower_bound}
    The maximum distinctness probability $\perr^{\max}$ satisfies Properties \ref{itm:positive_definite} and \ref{itm:data_processing}, and it is lower-bounded by
    \begin{equation}
        \perr^{\max}(\rho, \sigma) \geq \frac{1}{2} \| \rho - \sigma \|_{1}.
    \end{equation}
\end{restatable}

From this and \cref{lem:perr_upper}, we conclude that $\perr^{\max}(\rho, \sigma)$ lies in between $\frac{1}{2} \|\rho - \sigma\|_1$ and $\|\rho - \sigma\|_1$, i.e., the maximum distinctness probability is tightly characterised  by the trace distance.\footnote{One may wonder whether a similar result holds in gereralised probabilistic theories (GPTs), which allow defining security without relying on quantum theory~\cite{BHK2005}. Within GPTs, it is common to define a distance measure~$\delta_{\mathrm{dist}}$ based on the distinguishing advantage (see \cref{sec_critiques} for a brief explanation). Following the approach proposed here, one may instead consider a measure like $\perr^{\max}$ based on Axioms~\ref{itm:triangle_ineq}--\ref{itm:operational_bound}, which can also be formulated within GPTs. While we expect that $\delta_{\mathrm{dist}} \leq \perr^{\max}$, it could be that this relation is not tight. A security definition within the GPT framework based on $\perr^{\max}$ may then be strictly more stringent than one based on $\delta_{\mathrm{dist}}$ (see Definitions~1--3 in~\cite{HR2010}).}

\subsection{Security definition} \label{sec:securitydefinition}

In this section, we discuss the \emph{$\varepsilon$-security criterion}, introduced in \cite{BHLMO2005,RennerKoenig2005} and now a de facto standard in QKD. We present a generalized formulation that accommodates variable-length protocols with asymmetric abort conditions, which in turn recovers, as a special case, the commonly studied fixed-length protocols with symmetric abort conditions.

\begin{remark} 
  The existing literature is not entirely consistent regarding the definition of the security parameter~$\varepsilon$. However, the definitions differ by, at most, a factor of~$2$, which is irrelevant in practice. The reason is that a protocol with security $\varepsilon$ can be turned into one with improved security $\varepsilon/2$ by slightly adapting the protocol parameters. Crucially, such an adaption has only a small impact on the resource requirements. 
\end{remark}

\subsubsection{The $\boldsymbol{\varepsilon}$-security criterion}
\label{subsec:security_criteria}

Let $\Omega_l$ be the event that a key of size $l$ was produced, let $K_{A}^{l}$ and $K_{B}^{l}$ represent Alice's and Bob's (classical) keys of length $l$ (i.e., registers with alphabet $\{0,1\}^l$), and let $\Efin$ be (possibly quantum) information the adversary may have about these keys (see also Fig.~\ref{fig_QKD_protocol}). We denote the state of these systems after the execution of a given QKD protocol, conditioned on $\Omega_l$, by $\rho^{\mathrm{real}}_{K_{A}^{l}K_{B}^{l}\Efin|\Omega_l}$ and call it the \emph{real state}.

Furthermore, let us define a corresponding \emph{ideal state} as the hypothetical joint state these systems would admit if the QKD protocol were perfectly secure. By this we mean that the following two properties hold:
\begin{itemize}
    \item The two keys $K_A^{l}$ and $K_B^{l}$ are identical.
    \item The key $K_A^{l}$ is uniformly distributed and uncorrelated with Eve's system $\Efin$.
\end{itemize}
Mathematically, we capture these properties by the conditions
\begin{equation}
    \rhoid_{K_{A}^{l} K_{B}^{l}|\Omega_l} = \sum_{k\in\{0,1\}^l} p^{(l)}_{k} \ketbra{k, k}{k, k}_{K_{A}^{l} K_{B}^{l}}
\end{equation}
and
\begin{equation}
    \rhoid_{K_{A}^{l}\Efin|\Omega_l}= \frac{\mathbb{1}_{K_{A}^{l}}}{2^l} \otimes \rho^{\mathrm{real}}_{\Efin|\Omega_l}  = \tau_{K_A^l}\otimes \rho^{\mathrm{real}}_{\Efin|\Omega_l},
\end{equation}
where $p_k^{(l)}$ for ${k\in\{0,1\}^l}$ is a probability distribution over the set of possible keys, and $\smash{\tau_{K_{A}^l}} \coloneqq \smash{2^{-l} \mathds{1}_{K_A^l}}$ is the maximally mixed state on $K_{A}^l$. Combining these two conditions, we find that the ideal state conditioned on $\Omega_l$ is
\begin{equation}\label{eq_idealstate}    \rhoid_{K_{A}^{l}K_{B}^{l}\Efin|\Omega_l} = \underbrace{\left( \frac{1}{2^l}\sum_{k\in\{0,1\}^l} \ketbra{k,k}_{K_{A}^{l} K_{B}^{l}} \right)}_{\eqcolon\tau_{K_{A}^{l}K_{B}^{l}}} \otimes \rho^{\mathrm{real}}_{\Efin|\Omega_{l}}, 
\end{equation}
where $\tau_{K_{A}^{l} K_{B}^{l}}$ is the maximally correlated state of the two keys $K_{A}^{l}$ and $K_{B}^{l}$.

\begin{remark} 
    In the literature, the ideal state is sometimes allowed to have a different marginal on $\Efin$ than the real state $\rho^{\mathrm{real}}$, i.e., the ideal state is taken to be of the form
    \begin{equation}
        \rhoid_{K_{A}^{l}K_{B}^{l}\Efin|\Omega_l} = \tau_{K_{A}^{l}K_{B}^{l}} \otimes \sigma_{\Efin|\Omega_{l}},
    \end{equation}
    where $\sigma_{\Efin|\Omega_{l}}$ may differ from $\smash{\rho^{\mathrm{real}}_{\Efin|\Omega_{l}}}$. As shown in \cite[Footnote~24]{Portmann_2022}, this leads to the same security criterion, up to a factor of~$2$ in the security parameter~$\varepsilon$ (see also the remark above). Here we stick to~\eqref{eq_idealstate}, which has more natural properties under composition (see Section~\ref{sec_composability}).
\end{remark}

In \cref{lem:perr_upper}, we showed that twice the trace distance provides an upper bound on the probability that two states differ, 
\begin{equation} \label{eq_PdiffTD}
    \perr(\rho^{\text{real}}, \rhoid) \leq \| \rho^{\text{real}} - \rhoid \|_{1}.
\end{equation}
The central idea of the security definition (stated below as \cref{def:security}) is thus to require that the right hand side of this inequality is bounded by~$\varepsilon$.

Our considerations so far were conditioned on the event $\Omega_l$ that a key of a fixed length~$l$ was generated. To account for protocols that produce a key of varying length $l$, we consider final key registers $K_A$ and $K_B$ given by a direct sum of all possible key lengths, i.e., 
\begin{equation}
    \mathcal{H}_{K_A} = \bigoplus_{l \geq 0} \mathcal{H}_{K_A^l} \quad \mathrm{and} \quad \mathcal{H}_{K_B} = \bigoplus_{l \geq 0} \mathcal{H}_{K_B^l}.
\end{equation}
Note that Eve may always provoke an abort of the protocol, which we would regard as the generation of a key of length $l=0$. 

For the purpose of generality, one may also include the possibility of an asymmetric abort, meaning that Alice produces a key whereas Bob doesn't, or vice versa. This can be achieved by replacing the key length~$l$ by two values, $l_A$ and $l_B$, corresponding to Alice and Bob's key lengths, respectively, in the expressions above. We then either have $l_A = l_B \neq 0$ (i.e., both parties accept and produce a key) or $l_A = 0$, or $l_B = 0$ (i.e., either or both of the parties abort). Defining $\mathcal{K} \coloneqq \{(l_A, l_B) : l_A = l_B \lor l_A = 0 \lor l_B = 0\}$,  the real state can thus be written as
\begin{equation} \label{eq:real_state_asymmetric_abort}
    \rho_{K_A K_B \Efin}^{\mathrm{real}} = \sum_{(l_A, l_B) \in \mathcal{K}} \mathrm{Pr}[\Omega_{l_A,l_B}]\,\rho_{K_A^{l_A} K_B^{l_B} \Efin|\Omega_{l_A,l_B}}^\mathrm{real},
\end{equation}
where $\Omega_{l_A,l_B}$ denotes the event that keys of lengths $l_A$ and $l_B$ were produced. We define the ideal state as the real state in which the key is replaced with a perfectly secret key. Formally, we have
\begin{equation}
    \rho_{K_A K_B \Efin}^\mathrm{ideal} \coloneqq \left(\mathcal{R}_{K_A K_B} \otimes \mathcal{I}_{\Efin}\right) \left[ \rho_{K_A K_B \Efin}^\mathrm{real} \right],
\end{equation}
where $\mathcal{R}$ is the \emph{key-replacer map},  which maps any input $\sigma_{K_A K_B}$ to
\begin{equation} \label{eq:key_replacer_map}
\begin{aligned}
    \mathcal{R}_{K_A K_B}[\sigma_{K_A K_B}]
    &= \mathrm{Pr}_\sigma[\Omega_{0,0}] \ketbra{\varnothing}_{K_A^0} \otimes \ketbra{\varnothing}_{K_B^0} \\
    &\hspace{5pt}+ \sum_{l > 0} \begin{aligned}[t]
        &\left( \mathrm{Pr}_{\sigma}[\Omega_{l,0}] \tau_{K_A^l} \otimes \ketbra{\varnothing}_{K_B^0} \right.\\
        &\;\;+ \mathrm{Pr}_{\sigma}[\Omega_{0,l}] \ketbra{\varnothing}_{K_A^0} \otimes \tau_{K_B^l} \\
        &\;\;\left. +\, \mathrm{Pr}_{\sigma}[\Omega_{l,l}] \tau_{K_A^l K_B^l}  \right),
    \end{aligned}
\end{aligned}
\end{equation}
where $\varnothing$ denotes the bitstring of length zero.

We are now ready to state the QKD security definition.

\begin{definition}[$\varepsilon$-security] \label{def:security}
    Let $\varepsilon \in [0, 1]$. We call a QKD protocol $\varepsilon$-secure if the final state of the protocol, $\rho^{\mathrm{real}}_{K_{A}K_{B}\Efin}$, satisfies
    \begin{equation} \label{eq:security_definition}
        \left\| \rho^{\mathrm{real}}_{K_{A}K_{B}\Efin} - \rhoid_{K_{A}K_{B}\Efin}\right\|_{1} \leq \varepsilon.
    \end{equation}
\end{definition}

Because of \eqref{eq_PdiffTD}, this implies
\begin{equation} \label{eq:security_perr}
  \perr^{\max}\left( \rho^{\mathrm{real}}_{K_{A}K_{B}\Efin}, \rhoid_{K_{A}K_{B}\Efin}\right) \leq \varepsilon.
\end{equation}
Hence, following the considerations of \cref{sec:trace-distance}, one can interpret the security parameter~$\varepsilon$ as the maximum probability that the (real) key generated by the protocol deviates from an (ideal) perfectly secure key. This answers question~\textbf{Q} posed in the introduction.

The security definition given here is more general than the one typically found in the literature \cite{Portmann_2022}. We now show that the latter arises as a special case.

\subsubsection{QKD with symmetric abort and fixed key length}

For simplicity, it is common to assume that the abort always occurs symmetrically, i.e., either Alice and Bob \emph{both} accept or they \emph{both} abort.\footnote{While this condition cannot be enforced (Eve may always corrupt the last message exchanged for confirming acceptance), it is a useful idealisation that allows focussing on the core concepts of the security definition.} In this case, the state in \cref{eq:real_state_asymmetric_abort} simplifies to

\begin{equation}
    \rho_{K_A K_B \Efin}^{\mathrm{real}} = \sum_{l} \mathrm{Pr}[\Omega_{l}]\,\rho_{K_A^lK_B^l\Efin|\Omega_{l}}^\mathrm{real},
\end{equation}
where $\Omega_l$ denotes the event of a key of length $l$ being produced. The abort event corresponds to the event $l=0$.

In the symmetric abort case, it often makes sense to split the security criterion in \cref{eq:real_state_asymmetric_abort} into two parts, called \emph{correctness} and \emph{secrecy}.

\begin{definition}[$\varepsilon$-correctness]
\label{def: trace distance criterion}
    Let $\varepsilon \in [0, 1]$. We call a QKD protocol $\varepsilon$-correct if the final state of the protocol, $\rho^{\mathrm{real}}_{K_{A}K_{B}}$, satisfies
        \begin{equation}
            \left\| \rho_{K_{A}K_{B}}^{\mathrm{real}} - \hat{\rho}_{K_{A}K_{B}} \right\|_1 \leq \varepsilon,
        \end{equation}
        where $\hat{\rho}_{K_A K_B}$ is equal to $\rho_{K_A K_B}^{\mathrm{real}}$ but with Bob's key replaced by a copy of Alice's key.
\end{definition}

\begin{remark}
As shown in \cref{lem:correctness_alt_form}, the correctness condition can be rewritten in the following equivalent form
\begin{equation}
    \mathrm{Pr}[K_{A} \neq K_{B} \land l > 0] \leq \frac{1}{2} \varepsilon,
\end{equation}
which is frequently used in the literature (up to a factor of~$2$, as remarked earlier).
\end{remark}
\begin{definition}[$\varepsilon$-secrecy]
    \label{def: var length trace distance criterion}
    Let $\varepsilon \in [0, 1]$. We call a QKD protocol $\varepsilon$-secret according to Alice if the final state of the protocol, $\rho^{\mathrm{real}}_{K_{A}\Efin}$, satisfies
    \begin{equation} \label{eq:secrecy_definition}
    \left\| \rho^{\mathrm{real}}_{K_{A}\Efin} - \rhoid_{K_{A}\Efin}\right\|_{1} \leq \varepsilon.
    \end{equation}
\end{definition}

\begin{remark}    
The secrecy definition in \cref{eq:secrecy_definition} can be rewritten as
\begin{equation}
    \sum_{l} \mathrm{Pr}[\Omega_{l}] \left\| \rho^{\mathrm{real}}_{K_{A}^{l}\Efin|\Omega_{l}} - \tau_{K_A^l} \otimes \rho^{\mathrm{real}}_{\Efin|\Omega_{l}} \right\|_{1} \leq \varepsilon,
\end{equation}
where $\Omega_{l}$ denotes the event that a key of length $l$ was produced and $\tau_{K_A^l} = \frac{\mathds{1}}{2^{l}}$ denotes the maximally mixed state on $K_A^l$. 
\end{remark}

\begin{remark}[Fixed-length security]
An important special case of the security definition is when the key length $l$ can only take one of two values $l \in \{0, L\}$, where $l = 0$ corresponds to the protocol aborting whereas $l = L$ corresponds to the protocol accepting. Then $\varepsilon$-secrecy as per \cref{def: var length trace distance criterion} is captured by the simplified criterion
\begin{equation} \label{eq:fixed-length-secrecy}
    \mathrm{Pr}\left[\Omega\right] \left\| \rho^{\mathrm{real}}_{K_{A}^L\Efin|\Omega} - \tau_{K_{A}^L} \otimes \rho^{\mathrm{real}}_{\Efin|\Omega} \right\|_{1} \leq \varepsilon,
\end{equation}
where $\Omega$ denotes the event of the protocol accepting, and $\tau_{K_{A}^L}$ is the maximally mixed state on $K_{A}^L$. Note also the prefactor $\mathrm{Pr}[\Omega]$ in \cref{eq:fixed-length-secrecy}, which is important. It accounts for attack strategies that allow Eve to learn the key but have a negligible success probability (e.g., Eve could correctly guess all the basis choices in a BB84 protocol). Sometimes this prefactor is ignored and it is incorrectly stated that \emph{conditioned on accepting}, the key is secure (see \cite{BBBMR2006} for a  discussion of this issue). The correct statement is that the probability of  producing an insecure key \emph{and accepting} is small.
\end{remark}

The secrecy criterion is defined with respect to Alice’s key $K_A$. An analogous formulation can be made for Bob’s key $K_B$; however, if the correctness criterion holds so that $K_A = K_B$, then the secrecy of $K_B$ follows directly from that of $K_A$. Hence, security can be derived from the combination of secrecy and correctness. This relationship is captured by the following lemma, whose proof is given in \cite[Theorem~2]{Portmann_2022}.

\begin{lemma}[$\varepsilon$-security from secrecy and correctness] 
    Consider a QKD protocol that is $\varepsilon_{\textup{cor}}$-correct and $\varepsilon_{\textup{sec}}$-secret according to Alice for $\varepsilon_{\textup{cor}},\varepsilon_{\textup{sec}}\in [0,1]$. Then the final state of the protocol, $\rho^{\mathrm{real}}_{K_{A}K_{B}\Efin}$, is $\varepsilon$-secure for $\varepsilon=\varepsilon_{\textup{sec}}+\varepsilon_{\textup{cor}}$, i.e.,
    \begin{equation}
        \left\lVert \rho^{\mathrm{real}}_{K_AK_B\Efin} - \rhoid_{K_AK_B\Efin}\right\rVert_{1} \leq \varepsilon.
    \end{equation}
\end{lemma}

In practice, it is sometimes advantageous to treat correctness and secrecy as separate criteria. The reason is that the two types of errors can have very different consequences when the key is used in applications. For instance, consider one-time pad encryption of messages from Alice to Bob using the key pair $(K_{A},K_{B})$. If $K_A \neq K_B$, Bob may decrypt incorrectly. This happens with probability at most $\varepsilon_{\mathrm{cor}}$. Independently of this, the encrypted message may leak to Eve with probability at most $\varepsilon_{\mathrm{sec}}$. Since these two implications are quite different, one may assign different values to $\varepsilon_{\mathrm{cor}}$ and $\varepsilon_{\mathrm{sec}}$, rather than merging them into a single security parameter~$\varepsilon$.

\subsubsection{How to choose $\boldsymbol{\varepsilon}$?}

To determine an acceptable value for the security parameter~$\varepsilon$, one may draw on practices from safety-critical engineering---such as the design of bridges or nuclear reactors---where acceptable failure probabilities are chosen in light of the potential consequences. Concretely, a security breach can be associated with a potential cost $C_{\mathrm{failure}}$. For example, $C_{\mathrm{failure}}$ may be thought of as the insured amount, i.e., the loss one would want an insurance policy to cover in the event that an adversary successfully breaks the QKD scheme. Since the security parameter~$\varepsilon$ provides an upper bound on the probability of such a security failure, $\varepsilon$ likewise upper-bounds the probability that the cost $C_{\mathrm{failure}}$ is realised. 

The effective insurance cost scales approximately as $C_{\mathrm{insurance}} \approx \varepsilon \, C_{\mathrm{failure}}$: it becomes larger as $\varepsilon$ increases. Conversely, making $\varepsilon$ smaller incurs a resource cost $C_{\mathrm{resource}}$, for example through increased communication requirements. The choice of~$\varepsilon$ is therefore a trade-off between $C_{\mathrm{resource}}$ and $C_{\mathrm{insurance}}$. A natural heuristic is to select~$\varepsilon$ near the point where these two costs balance (see Fig.~8 of \cite{RennerWolf2023} for an illustration).

Since QKD keys are typically used within larger cryptographic systems, the choice of $\varepsilon$ should also be consistent with the failure probabilities of other components. Selecting $\varepsilon$ to be significantly smaller than those of the surrounding system offers little benefit and may impose unnecessary resource costs.

It is also essential to account for the repeated use of a QKD protocol. Note that the parameter~$\varepsilon$ bounds the probability of a security failure in a single execution of the protocol. If the protocol is run $n$ times to generate $n$ independent keys, the total failure probability is upper-bounded by $n\varepsilon$ (we will show this explicitly in~\cref{sec_composability} below). Consequently, to ensure that all $n$ keys are secure with high probability, $\varepsilon$ must be chosen so that $n\varepsilon$ remains below a desired threshold.

\subsubsection{Completeness}

\Cref{def:security} is sufficient to guarantee the security of the generated key. However, there is one aspect which it fails to capture: A QKD protocol which always aborts is trivially secure. Clearly, such a protocol, although secure, would be useless. One thus generally also requires that for certain (noisy) channels---ideally those that model the real-world communication link between Alice and Bob when Eve is passive---the protocol accepts with reasonably high probability and produces a key. This property is called the \emph{completeness} (or sometimes robustness) of a QKD protocol. 

Unlike the secrecy condition, the completeness of a protocol can be experimentally verified. Concretely, the acceptance probability can be estimated by repeatedly running the protocol on a given communication channel and recording the frequency of acceptance versus abortion. 

\subsection{Composability} \label{sec_composability}

One is rarely interested in the security of a QKD protocol in isolation. Rather, one wants to use it within an application that includes other cryptographic routines (e.g., one-time-pad encryption, as we have seen before, or other instances of QKD protocols). It is thus crucial to ensure that the security definitions remain valid in such scenarios.

This is addressed by the notion of \emph{composability}~\cite{Pfitzmann_2000, Pfitzmann_2001,Canetti_2001, Maurer_2011}. The idea is to allow for modularity, i.e., security is proved separately for each individual piece of a composite cryptographic system. Composability then corresponds to the requirement that the individual security statements imply security of the overall system.  \Cref{def:security} is compatible with this composability requirement. We illustrate this in the following with two concrete examples, corresponding to sequential and parallel composition of QKD protocols. As we will see, the failure probabilities $\varepsilon$ of the different components simply add up if they are used together. 

\subsubsection{Sequential composition}
Consider two QKD protocols that are run in sequence. We can model the first QKD protocol as a CPTP map $\mathcal{E}$ taking some initial state $\rho^{(0)}$ to some state $\rho^{(1,\mathrm{real})} = \mathcal{E}[\rho^{(0)}]$. Similarly, the second QKD protocol is again a CPTP map $\mathcal{F}$ mapping the state $\rho^{(1,\mathrm{real})}$ to some final state $\rho^{(2,\mathrm{real})} = \mathcal{F}[\rho^{(1,\mathrm{real})}]$. Let us denote by $\rho^{(1,\mathrm{ideal})}$ the ideal state of the first QKD protocol, i.e., $\rho^{(1,\mathrm{ideal})} = \mathcal{R}[\rho^{(1,\mathrm{real})}]$, where $\mathcal{R}$ is the key-replacer map from~\cref{eq:key_replacer_map}.
Similarly, define $\rho^{(2,\mathrm{ideal})} = \mathcal{R}\circ\mathcal{F}[\rho^{(1,\mathrm{ideal})}]$. Because the maps $\mathcal{R}$ only act on the output keys of the QKD protocols, we have $\rho^{(2,\mathrm{ideal})} = \mathcal{R}'[\rho^{(2,\mathrm{real})}]$, where $\mathcal{R}'$ is the map that replaces both keys.\footnote{This is not true if one wants to use the key of the first QKD protocol to authenticate the classical communication in the second QKD protocol. In this case, a more complicated definition of the ideal state $\rho^{(2,\mathrm{ideal})}$ is necessary. We do not cover this more elaborate form of composition here, but instead refer to the literature \cite{Portmann_2014_key_recycling, Maurer_2011}.} Suppose that the two QKD protocols have security parameter $\varepsilon_1$ and $\varepsilon_2$. The composed protocol then satisfies
\begin{equation}
\begin{aligned}
    \left\| \mathcal{F} \circ \mathcal{E}[\rho^{(0)}] - \rho^{(2,\mathrm{ideal})} \right\|_1 \leq& \left\| \mathcal{F}\circ\mathcal{E}[\rho^{(0)}] - \mathcal{F}[\rho^{(1,\mathrm{ideal})}] \right\|_1 + \left\| \mathcal{F}[\rho^{(1, \mathrm{ideal})}] - \rho^{(2,\mathrm{ideal})} \right\|_1 \\
    \leq& \left\| \mathcal{E}[\rho^{(0)}] - \rho^{(1,\mathrm{ideal})} \right\|_1 + \left\| \mathcal{F}[\rho^{(1,\mathrm{ideal})}] - \rho^{(2,\mathrm{ideal})} \right\|_1 \\
    \leq& \varepsilon_1 + \varepsilon_2,
\end{aligned}
\end{equation}
where the first line follows by the triangle inequality, the second line follows because the trace distance can only decrease under CPTP maps, and for the third line we used that the individual protocols described by $\mathcal{E}$ and $\mathcal{F}$ are $\varepsilon_{1}$ and $\varepsilon_{2}$-secure, respectively. 

This argument can also be understood at a more intuitive level: The only way the composed protocol $\mathcal{F} \circ \mathcal{E}$ can fail is if either the first protocol $\mathcal{E}$ fails or if the second protocol $\mathcal{F}$ fails.
This illustrates that the security parameter of the composed protocol is upper-bounded by the sum of the security parameters of its individual building blocks.

Iterating this argument shows that if an $\varepsilon$-secure QKD protocol is executed $n$ times, the overall security parameter scales as $n \varepsilon$. Consequently, to ensure that all $n$ keys are secure, the product $n \varepsilon$, rather than only $\varepsilon$, must be small. Fortunately, common QKD protocols can achieve sufficiently small values of $\varepsilon$ with only a modest increase in resource consumption. Thus, given an upper bound on $n$ (e.g., corresponding to the lifetime of the device), the system can be designed so that $n \varepsilon$ remains acceptably small.

\begin{figure}[t]
    \centering
    \begin{tikzpicture}[
        green_box/.style={
            thick,
            text=black,
            draw=mygreen!90,
            rounded corners,
            minimum height=0.7cm,
            minimum width=0.9cm,
            fill=mygreen!60,
        },
        blue_box/.style={
            thick,
            text=black,
            draw=myblue!70,
            rounded corners,
            minimum height=0.7cm,
            minimum width=0.9cm,
            fill=myblue!50,
        },
    ]
        \node[] (Alice1) at (-6.25, 0) {Alice 1};
        \node[] (Bob1) at (-3.25, 0) {Bob 1};
        \node[] (Eve) at (0, 0) {Eve};
        \node[] (Bob2) at (3.25, 0) {Bob 2};
        \node[] (Alice2) at (6.25, 0) {Alice 2};

        \node[thick,draw=myorange!70,fill=myorange!40,rounded corners=0.5cm,minimum height=1cm, minimum width=14cm] (rho) at (0, -0.8) {$\rho_{A_1 B_1 E B_2 A_2}$};

    
        \node[blue_box] (M1) at ([yshift=-1.9cm] Alice1.south) {$\mathcal{M}$};
        \draw[->,thick,>=stealth] (rho.south-|M1.north) -- (M1.north);
        
        \node[blue_box] (N1) at ([yshift=-1.9cm] Bob1.south) {$\mathcal{N}$};
        \draw[->,thick,>=stealth] (rho.south-|N1.north) -- (N1.north);
    
        \node[blue_box] (M2) at ([yshift=-1.9cm] Alice2.south) {$\mathcal{M}$};
        \draw[->,thick,>=stealth] (rho.south-|M2.north) -- (M2.north);
        
        \node[blue_box] (N2) at ([yshift=-1.9cm] Bob2.south) {$\mathcal{N}$};
        \draw[->,thick,>=stealth] (rho.south-|N2.north) -- (N2.north);
    
        \node[green_box,minimum width=4cm] (Eq) at ([yshift=-3.3cm] Eve) {$\mathcal{E}^{\mathrm{quant}}$};
        \draw[thick,->,>=stealth] (rho.south-|Eq.north) -- (Eq.north);
        \draw[thick,->,>=stealth] ([xshift=+0.2cm] M1.south) -- ([xshift=+0.2cm] M1.south|-Eq.west) -- (Eq.west);
        \draw[thick,->,>=stealth] ([xshift=-0.2cm] M2.south) -- ([xshift=-0.2cm] M2.south|-Eq.east) -- (Eq.east);
    
        \draw[->,thick,>=stealth] ([xshift=-0.5cm] Eq.north) -- ([xshift=-0.5cm] N1.east-|Eq.north) -- (N1.east);
        \draw[->,thick,>=stealth] ([xshift=+0.5cm] Eq.north) -- ([xshift=+0.5cm] N2.west-|Eq.north) -- (N2.west);
    
        \node[blue_box] (PA1) at ([yshift=-2cm] M1.south) {$\mathcal{P}^{\mathrm{Alice}}$};
        \draw[thick,->,>=stealth] (M1.south) -- (PA1.north);
        \node[green_box] (EcA1) at ([xshift=+1.5cm] PA1.center) {$\mathcal{E}^{\mathrm{auth}}$};
        \draw[->,thick,>=stealth] (PA1.east) -- (EcA1.west);
        \draw[->,thick,>=stealth] ([yshift=0.2cm] EcA1.east) -- ([xshift=-0.1cm,yshift=0.2cm] EcA1.east-|Eq.south);
        \node[blue_box] (PB1) at ([yshift=-3cm] N1.south) {$\mathcal{P}^{\mathrm{Bob}}$};
        \draw[->,thick,>=stealth] (EcA1.east) -- ([xshift=-0.2cm] EcA1.east-|PB1.north) -- ([xshift=-0.2cm] PB1.north);
        \draw[->,thick,>=stealth] (N1.south) -- (PB1.north);
        \node[green_box] (EcB1) at ([xshift=-1.5cm] PB1.center) {$\mathcal{E}^{\mathrm{auth}}$};
        \draw[thick,->,>=stealth] (PB1.west) -- (EcB1.east);
        \draw[thick,->,>=stealth] (EcB1.west) -- (EcB1.west-|PA1.south);    
        \draw[thick,->,>=stealth] (EcB1.south) -- ([yshift=-0.5cm] EcB1.south) -- ([xshift=-0.1cm,yshift=-0.5cm] EcB1.south-|Eq.center);
    
        \node[blue_box] (PA2) at ([yshift=-2cm] M2.south) {$\mathcal{P}^{\mathrm{Alice}}$};
        \draw[thick,->,>=stealth] (M2.south) -- (PA2.north);
        \node[green_box] (EcA2) at ([xshift=-1.5cm] PA2.center) {$\mathcal{E}^{\mathrm{auth}}$};
        \draw[->,thick,>=stealth] (PA2.west) -- (EcA2.east);
        \draw[->,thick,>=stealth] ([yshift=0.2cm] EcA2.west) -- ([xshift=+0.1cm,yshift=0.2cm] EcA2.west-|Eq.south);
        \node[blue_box] (PB2) at ([yshift=-3cm] N2.south) {$\mathcal{P}^{\mathrm{Bob}}$};
        \draw[->,thick,>=stealth] (EcA2.west) -- ([xshift=+0.2cm] EcA2.west-|PB2.north) -- ([xshift=+0.2cm] PB2.north);
        \draw[->,thick,>=stealth] (N2.south) -- (PB2.north);
        \node[green_box] (EcB2) at ([xshift=+1.5cm] PB2.center) {$\mathcal{E}^{\mathrm{auth}}$};
        \draw[thick,->,>=stealth] (PB2.east) -- (EcB2.west);
        \draw[thick,->,>=stealth] (EcB2.east) -- (EcB2.east-|PA2.south);
        \draw[thick,->,>=stealth] (EcB2.south) -- ([yshift=-0.5cm] EcB2.south) -- ([xshift=+0.1cm,yshift=-0.5cm] EcB2.south-|Eq.center);
    
        \node[blue_box] (KA1) at ([yshift=-2.5cm] PA1.south) {$\mathcal{K}^{\mathrm{Alice}}$};
        \draw[thick,->,>=stealth] (PA1.south) -- (KA1.north);
        \node (KeyA1) at ([yshift=-0.8cm] KA1.south) {$K_{A_{1}}$};
        \draw[thick,->,>=stealth] (KA1.south) -- (KeyA1.north);
        \node[blue_box] (KB1) at (KA1.center-|PB1.south) {$\mathcal{K}^{\mathrm{Bob}}$};
        \draw[thick,->,>=stealth] (PB1.south) -- (KB1.north);
        \node (KeyB1) at ([yshift=-0.8cm] KB1.south) {$K_{B_{1}}$};
        \draw[thick,->,>=stealth] (KB1.south) -- (KeyB1.north);
    
        \node[blue_box] (KA2) at ([yshift=-2.5cm] PA2.south) {$\mathcal{K}^{\mathrm{Alice}}$};
        \draw[thick,->,>=stealth] (PA2.south) -- (KA2.north);
        \node (KeyA2) at ([yshift=-0.8cm] KA2.south) {$K_{A_{2}}$};
        \draw[thick,->,>=stealth] (KA2.south) -- (KeyA2.north);
        \node[blue_box] (KB2) at (KA2.center-|PB2.south) {$\mathcal{K}^{\mathrm{Bob}}$};
        \draw[thick,->,>=stealth] (PB2.south) -- (KB2.north);
        \node (KeyB2) at ([yshift=-0.8cm] KB2.south) {$K_{B_{2}}$};
        \draw[thick,->,>=stealth] (KB2.south) -- (KeyB2.north);
    
        \node (Ehat) at (Eq.center|-KeyA1.center) {$\Efin$};
        \draw[thick,->,>=stealth] (Eq.south) -- (Ehat.north);

        \begin{scope}[on background layer]
            \draw[draw=lightgray!80!black,fill=lightgray!40,rounded corners]
                ([xshift=-0.3cm, yshift=+0.2cm] Alice1.north west) --
                ([xshift=+0.3cm, yshift=+0.2cm] Bob1.north east) --
                ([xshift=+0.3cm, yshift=-0.1cm] Bob1.north east|-KeyB1.south east) --
                ([xshift=-0.3cm, yshift=-0.1cm] Alice1.north west|-KeyA1.south west) --
                cycle;
            \draw[draw=lightgray!80!black,fill=lightgray!40,rounded corners]
                ([xshift=-0.3cm, yshift=+0.2cm] Bob2.north west) --
                ([xshift=+0.3cm, yshift=+0.2cm] Alice2.north east) --
                ([xshift=+0.3cm, yshift=-0.1cm] Alice2.north east|-KeyA2.south east) --
                ([xshift=-0.3cm, yshift=-0.1cm] Bob2.north west|-KeyB2.south west) --
                cycle;
        \end{scope}
    \end{tikzpicture}
    \caption{\textbf{Parallel composition of two QKD protocol executions.} Each of the gray boxes corresponds to the execution of a QKD protocol as in~\cref{fig_QKD_protocol}. By the assumptions of the protocol, the operations of the honest parties (blue boxes) factorize between the two executions of the protocol, whereas Eve's actions (green boxes) can introduce correlations between the two protocols.
    }
    \label{fig:parallel_composition}
\end{figure}
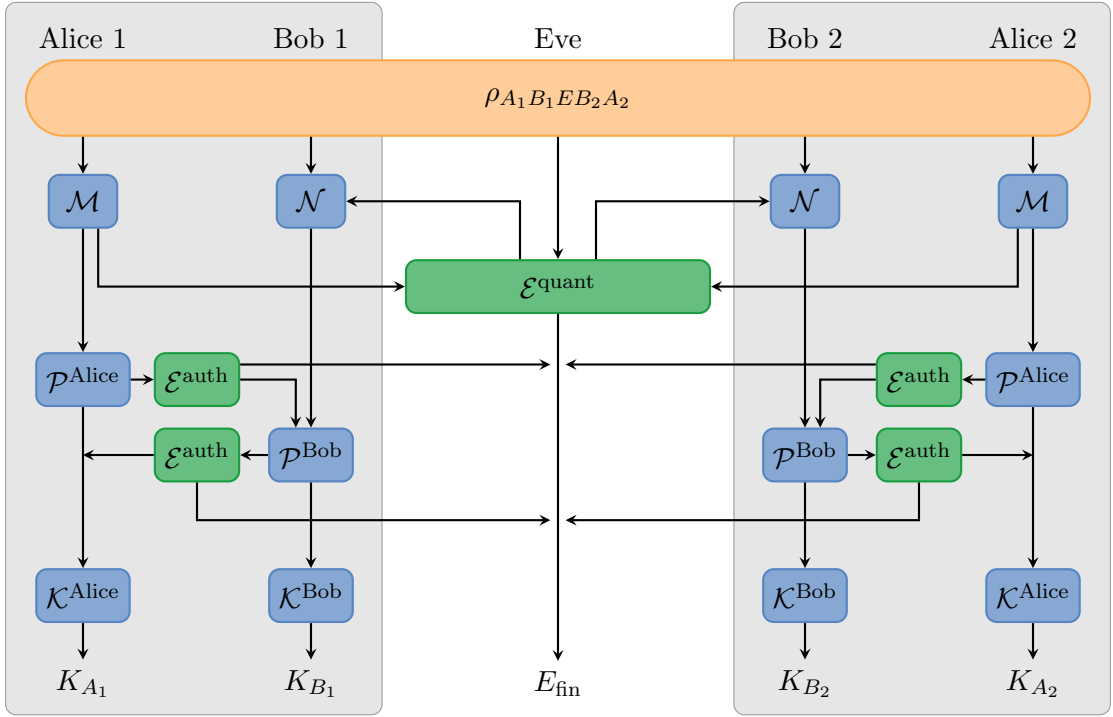

\subsubsection{Parallel composition}
Similarly, one can consider the parallel execution of two instances of a given QKD protocol. This situation is sketched in Fig.~\ref{fig:parallel_composition}. Let $\Omega_{1}$ denote the event of the first protocol accepting and $\Omega_{2}$ denote the event of the second protocol accepting. Then the event of both protocols accepting may be written as $\Omega \coloneqq \Omega_{1} \cap \Omega_{2}$. Given that the protocols are each $\varepsilon$-secret, we have for the final key $K_{A_{1}}K_{A_{2}}$ that
\begin{equation}
\begin{aligned}
    \mathrm{Pr}[\Omega] &\trdist{\rho_{K_{A_1}K_{A_2}\Efin | \Omega}}{\tau_{K_{A_1} K_{A_2}} \otimes \rho_{\Efin | \Omega}}\\
    &\hspace{10pt}= \trdist{\rho_{K_{A_1}K_{A_2}\Efin \land \Omega}}{\tau_{K_{A_1}} \otimes \tau_{K_{A_2}} \otimes \rho_{\Efin \land \Omega}} \\
    &\hspace{10pt}\leq\trdist{\rho_{K_{A_{1}}K_{A_{2}}\Efin \land \Omega}}{\tau_{K_{A_1}} \otimes \rho_{K_{A_2} \Efin \land \Omega}}\\
    &\hspace{25pt}+\trdist{\tau_{K_{A_1}}\otimes \rho_{K_{A_2} \Efin \land \Omega}}{\tau_{K_{A_1}} \otimes \tau_{K_{A_2}} \otimes \rho_{\Efin \land \Omega}}\\
    &\hspace{10pt}\leq \trdist{\rho_{K_{A_{1}}K_{A_{2}}\Efin \land \Omega_{1}}}{\tau_{K_{A_1}} \otimes \rho_{K_{A_2} \Efin \land \Omega_{1}}} \\
    &\hspace{25pt}+\trdist{\tau_{K_{A_1}} \otimes \rho_{K_{A_2} \Efin \land \Omega_{2}}}{\tau_{K_{A_1}} \otimes \tau_{K_{A_2}} \otimes \rho_{\Efin \land \Omega_{2}}} \\
    &\hspace{10pt}\leq \varepsilon + \trdist{\rho_{K_{A_2} \Efin \land \Omega_{2}}}{\tau_{K_{A_2}} \otimes \rho_{\Efin \land \Omega_{2}}} \\
    &\hspace{10pt}\leq 2\varepsilon,
\end{aligned}
\end{equation}
where we first used the triangle inequality. The second inequality follows because the trace distance can only decrease under completely positive and trace non-increasing maps. For the third inequality, we used that the second QKD protocol can be seen as part of Eve's attack on the first QKD protocol. Finally, we used security of the second QKD protocol. The total key $K_{A_1} K_{A_2}$ is thus $2 \varepsilon$-secret. An analogous bound holds for correctness. Hence, in parallel composition, just like in sequential composition, we can simply add the individual errors to derive the security of the composite protocol.

\subsubsection{Probability of events} \label{sec:eventsincrease}

Another important consequence of the $\varepsilon$-security criterion is that the likelihood of any event cannot change by more than~$\varepsilon$ if one replaces a hypothetical perfect key generation device with an actual QKD protocol. To make this more precise, let $\rho^\mathrm{real}$ be the state of the ``real world'' where an actual QKD protocol has been executed. Furthermore, let $\rho^\mathrm{ideal}$ be the state of an ``ideal world'', which is identical to the real world, except that the key generated by the QKD protocol is replaced by a perfectly secure key. 

Now consider using the generated key for some cryptographic application, such as online banking. As any physical process, this transforms the state of the world according to a CPTP map $\mathcal{E}$. Suppose that there is a particular undesirable event---say, an adversary gaining unauthorized access to a bank account---which we denote by $\Omega$, and let $\Lambda$ be the POVM element corresponding to this event.
Then, the probabilities of $\Omega$ in the real and in the ideal world are related by
\begin{equation}
\begin{aligned}
    \mathrm{Pr}_{\rho^{\mathrm{real}}}[\Omega] =& \tr[\Lambda \mathcal{E}[\rho^\mathrm{real}]] \\
    =& \tr[\Lambda \mathcal{E}[\rho^\mathrm{ideal}]] + \tr[\Lambda \mathcal{E}[\rho^\mathrm{real} - \rho^\mathrm{ideal}]] \\
    \leq& \tr[\Lambda \mathcal{E}[\rho^\mathrm{ideal}]] + \frac{1}{2} \| \rho^\mathrm{real} - \rho^\mathrm{ideal}\|_1 \\
    \leq& \mathrm{Pr}_{\rho^\mathrm{ideal}}[\Omega] + \perr^{\max}(\rho^{\mathrm{ideal}}, \rho^{\mathrm{real}}),
\end{aligned}
\end{equation}
where we used the variational characterization of the trace distance, 
\begin{equation}
    \frac{1}{2} \| \rho - \sigma \|_1 = \max_{0 \leq \Lambda \leq \mathds{1}} \tr[\Lambda (\rho - \sigma)],
\end{equation} 
the fact that the trace distance can only decrease under CPTP maps, and \cref{lem:perr_lower_bound}.

This shows that using the key within the real world, described by $\rho^{\mathrm{real}}$, can increase the probability of the adverse event $\Omega$ by at most $\perr^{\max}(\rho^\mathrm{real}, \rho^\mathrm{ideal}) \leq \varepsilon$  compared to what would happen in the ideal world, described by $\rho^{\mathrm{ideal}}$. The security parameter $\varepsilon$ thus serves as a worst-case bound on the probability that imperfection arising from using a QKD-generated key---rather than a perfect key---can affect an application. This is compatible with the interpretation of~$\varepsilon$ as the maximum probability that the situations described by $\rho^\mathrm{real}$ and $\rho^\mathrm{ideal}$ are distinct.

\subsubsection{Composability frameworks}

Sequential and parallel composition, discussed above, are special cases of the broader principle of composability. To capture the composition of cryptographic protocols in general scenarios, several competing theoretical frameworks have been developed \cite{Pfitzmann_2000,Pfitzmann_2001,Canetti_2001,Maurer_2011,Maurer_2012,Broadbent_2023}. Some of these address classical cryptographic security \cite{Pfitzmann_2000,Pfitzmann_2001,Canetti_2001}, while others extend to quantum-mechanical protocols \cite{Maurer_2011,Maurer_2012}, and even to frameworks that encompass post-quantum or more general physical theories \cite{Portmann_2017}. Because the notion of $\varepsilon$-security discussed in this work rests on a set of well-motivated axioms, we expect it to be compatible with any reasonable framework for composability.

The framework of Abstract Cryptography (AC) \cite{Maurer_2011, Portmann_2022} is particularly well suited to encompass the axiomatic approach taken here. In AC, cryptographic systems are modelled in terms of abstract resources and converters, where converters transform one resource into another. Security is then defined by comparing a real resource with its ideal counterpart, with the difference measured by a (pseudo-)metric $\delta^{\mathrm{AC}}$. By choosing $\delta^{\mathrm{AC}} = \perr^{\max}$, the axioms of \cref{sec:trace-distance} directly imply the requirements of AC on the metric. In this sense, the definitions of~\cref{sec:securitydefinition} are compatible with the AC framework. 

\section{Critiques of the security criterion} \label{sec_critiques}

The $\varepsilon$-security criterion in QKD, \cref{def:security}, has been the subject of criticism, which has particularly focused on its reliance on the trace distance. Part of this criticism is directed at the so-called distinguishing paradigm \cite{Pfitzmann_2000,Pfitzmann_2001,Canetti_2001,Maurer_2011}, which is often cited as the main justification for using trace distance. Before addressing these concerns, we briefly review the distinguishing paradigm to set the stage.

The distinguishing paradigm envisions a hypothetical \emph{distinguishing game}: a distinguisher is given a black box that, depending on the value of a uniformly random bit~$B$, either executes the real protocol (e.g., for QKD) or simulates the corresponding idealised resource (such as a perfect key)~\cite{Pfitzmann_2000, Pfitzmann_2001, Canetti_2001, Maurer_2011}. The distinguisher’s goal is to output a guess~$\hat{B}$ for~$B$, based on interactions with the black box. One then defines the \emph{distinguishing advantage} as $2 {(\mathrm{Pr}[B=\hat{B}] - \frac{1}{2})}$, maximised over all possible strategies of the distinguisher. Note that ${\mathrm{Pr}[B=\hat{B}] = \frac{1}{2}}$ corresponds to a completely uninformed guess. Furthermore, it is a well-established result that, if the black box behaviour for $B=0$ and $B=1$ is described by quantum states $\rho^{\mathrm{real}}$ and $\rho^{\mathrm{ideal}}$, respectively, then the distinguishing advantage equals the trace distance $\frac{1}{2} \| \rho^{\mathrm{real}}- \rho^{\mathrm{ideal}}\|_1$ between these two states~\cite{Helstrom_1996}. Hence, upper-bounding this trace distance directly limits an adversary's possibility to see any difference between the (real) situation described by $\rho^{\mathrm{real}}$ and the (ideal) situation described by $\rho^{\mathrm{ideal}}$. 

\begin{enumerate}
    \item \textbf{Critique:}  In \cite{yuen_2016}, it has been argued that the bit $B$ should not be assumed to be uniformly distributed. The rationale was that one is always interacting with the real protocol, corresponding to a deterministic bit~$B$. More generally,~\cite{yuen_2016} raises the concern that the definition of the distinguishing advantage is specific to the distinguishing game and cannot be used to infer indistinguishability in other contexts.
    
    \textbf{Response:} If valid, such criticism would apply equally to classical cryptography, where the distinguishing-game argument is standard and widely regarded as unproblematic (see the references above). Moreover, there exist independent justifications for using the trace distance to define $\varepsilon$-security,  notably the one presented here, based on $\perr^{\max}$, which follows directly from general and well-justified axiomatic principles (discussed in \cref{sec:security-definition}).

    \item \textbf{Critique:} In \cite{yuen_2016}, it is proposed that one should bound the difference between two distributions (or quantum states in the case of QKD) for all arguments of the distribution individually. Concretely, for a key~$K$ of length $l$, this would mean that the probability $\mathrm{Pr}(K=k|E)$, from Eve's point of view, that $K$ takes any concrete value~$k$, should satisfy
    \begin{equation}
        \label{eq:yuen_criterion}
        \left|\mathrm{Pr}(K=k|E)  - 2^{-l} \right|\leq \varepsilon 2^{-l}
    \end{equation}
    for all $k\in\{0,1\}^l$.
    The $\varepsilon$-criterion does not imply~\cref{eq:yuen_criterion} as one can show with explicit examples~\cite{Hirota_2012}. Therefore, it is argued, the $\varepsilon$-criterion is not sufficient to ensure security.

    \textbf{Response:}  
    The criterion~\cref{eq:yuen_criterion} is \textit{sufficient} for security, in the sense that if~\cref{eq:yuen_criterion} is satisfied then the trace distance, and thus also $\perr$, is certainly bounded from above by $\varepsilon$. 
    However, one can easily see that it is a too pessimistic bound, and hence \emph{not necessary} to claim security. For instance, a distribution such that $\mathrm{Pr}(K=k|E) = (1-\varepsilon)2^{-l} + \varepsilon \delta_{k,0}$ is by construction equal to a uniform distribution, except with probability $\varepsilon$, but it does not satisfy~\cref{eq:yuen_criterion}. Consequently, being a too strong requirement, there's no need that the $\varepsilon$-security criterion implies \cref{eq:yuen_criterion}.

    \item \textbf{Critique:} Some authors attempt to justify the previous point by arguing that a single-number measure cannot be used to compare two distributions (density operators in the case of QKD), as it would only provide a bound on the mean value of the distribution~\cite{yuen_2016}.

    \textbf{Response:} Even if one accepted that mean values are inherently problematic, the argument is not conclusive. In particular, the claim that single-number measures can only bound mean values is not generally correct. For instance, the criterion proposed in~\cite{yuen_2016}, i.e., the bound~\cref{eq:yuen_criterion}, can be formulated as a single-number measure by multiplying~\cref{eq:yuen_criterion} with $2^l$ and taking the maximum of the left-hand side over all possible $k \in \{0,1\}^l$. Note, however, that the resulting measure violates the operational bound of Axiom~\ref{itm:operational_bound}, as one can verify using the example in the previous response.

    More generally, one may wonder whether there could be alternative, more stringent, security criteria than the $\varepsilon$-security criterion. These could be defined analogously to~\eqref{eq:security_perr}, i.e., as $\delta_{\mathrm{alternative}}(\rho^{\mathrm{real}}, \rho^{\mathrm{ideal}}) \leq \varepsilon$, where $\delta_{\mathrm{alternative}}$ is an alternative distance measure that is more sensitive than $\smash{\perr^{\max}}$. However, any such alternative measure would necessarily violate at least one of the Axioms~\ref{itm:triangle_ineq}--\ref{itm:operational_bound}. This is because $\perr$ was defined as the maximum (and, hence, most sensitive) measure satisfying these axioms.
        
    \item \textbf{Critique:} In practical implementations one usually has $\varepsilon \gg 2^{-l}$, as keys lengths are typically large (say, $l > 10^5$, whereas $\varepsilon \approx 10^{-10}$). This means that the parameters $\varepsilon$ and $2^{-l}$ differ significantly in their order of magnitude.  Therefore, even if the key~$K$ is $\varepsilon$-secure, the probability that $K$ takes a particular value~$k$ can still be relatively large, i.e., $\mathrm{Pr}(K=k) \gg 2^{-l}$~\cite{Yuen_2012}.
    
    \textbf{Response:} This consideration, although correct, is unproblematic. The reason is that, even though $\mathrm{Pr}(K=k) \gg 2^{-l}$, for an $\varepsilon$-secret key we still have $\mathrm{Pr}(K=k) \leq 2^{-l} + \varepsilon \ll 1$, i.e., the probability that Eve guesses correctly remains upper-bounded by (roughly)~$\varepsilon$ (because $2^{-l} \ll \varepsilon$). Note also that this is in perfect agreement with our considerations in \cref{sec:eventsincrease}.

\item \textbf{Critique:} In~\cite{abidin_2011} it was argued that when QKD is executed sequentially $n$ times, where the hash function used for authentication in each round is recycled in the next, the overall security parameter~$\varepsilon$ grows exponentially in~$n$. On this basis, the authors concluded that security degrades rapidly.

\textbf{Response:} The argument provided in~\cite{abidin_2011} merely yields an upper bound on the security parameter~$\varepsilon$ of the composed scheme. A refined treatment~\cite{Portmann_2014_key_recycling} demonstrates that $\varepsilon$ in fact grows only linearly in~$n$, consistent with our composability discussion in \cref{sec_composability}.

\item \textbf{Critique:} In~\cite{Bernstein2018} it is argued that the laws of physics, on which the security of QKD is based, also entail unavoidable interactions (such as through electromagnetic or gravitational fields) that make shielding between systems impossible. From this perspective, the laws of physics themselves would appear to preclude the very possibility of information-theoretic secrecy as captured by the $\varepsilon$-secrecy criterion.

 \textbf{Response:} This argument would indeed be justified within the framework of classical physics, where perfect shielding of fields is fundamentally impossible. However, the situation changes once these fields are treated as quantised. In fact, as argued in~\cite{RenesRenner2020}, if the laws of physics truly prohibited shielding, they would necessarily lead to unavoidable decoherence of any quantum information---thereby contradicting the foundational principle of quantum error correction. However, experiments in quantum information processing support the expectation that quantum error correction can suppress errors to arbitrarily small levels, i.e., below any desired threshold $\varepsilon > 0$.
\end{enumerate}

QKD is also often criticised for being unpractical and costly, in particular in comparison with post-quantum cryptography. A detailed discussion of these critiques can be found in \cite{NSAresponse}.

\section{Conclusion}

The secrecy of a key generated by QKD cannot be verified by direct experiments. The reason for this is that the legitimate parties, Alice and Bob, have no access to a potential adversary, Eve, so they cannot simply test whether she possesses any information about the key. Instead, secrecy must be inferred indirectly from observable quantities during the protocol---for example, the rate of observed bit errors. Drawing reliable conclusions from such observations is highly non-trivial and constitutes the central challenge of QKD security proofs.

At a high level, the structure of a QKD security proof is as follows: Under well-defined assumptions about Eve’s limitations (e.g., not having access to Alice and Bob’s laboratories), it is shown that the probability of her gaining significant information about the key is negligible (see~\cite{TTNKN2025} for a recent concrete example of such a proof). Making this precise requires two key ingredients: a formal model of the adversary's capabilities and a quantitative notion of secrecy.

The first ingredient, addressed in \cref{sec:assumptions}, involves specifying the resources used in a QKD protocol (such as trusted devices and authenticated classical communication) and stating the assumptions about Eve. This modelling step is where the greatest divergence between theory and implementation typically arises. If the assumptions made in the theoretical model are not satisfied in practice, the corresponding security proof does not apply. A well-known example is the photon-number splitting attack \cite{Huttner_1995}, which exploited a mismatch between the theoretical assumption of single-photon states (used to encode individual qubits) and the experimental use of weak coherent pulses. Identifying and bridging such gaps between theory and experiment has been the focus of significant effort in recent years and remains essential for achieving practically secure QKD~\cite{BSI_2023}. 

The second ingredient is the quantitative definition of security. In \cref{sec:security-definition}, we took an axiomatic approach to measure the probability that the situations described by two quantum states $\rho^{\mathrm{real}}$ and $\rho^{\mathrm{ideal}}$ are distinct, and we found that this probability is tightly upper-bounded by the trace distance between $\rho^{\mathrm{real}}$ and $\rho^{\mathrm{ideal}}$. Security of a QKD protocol is then defined by taking $\rho^{\mathrm{real}}$ to characterise the situation when running the protocol and $\rho^{\mathrm{ideal}}$ to be the situation with a perfectly secure key, and bounding their distance by a security parameter~$\varepsilon$. 
Hence, operationally, $\varepsilon$ can be interpreted as the maximum probability that the key fails to be secure. This interpretation also provides a meaningful basis for selecting practical security parameters and balancing security against other design constraints.

The $\varepsilon$-security criterion outlined in this document has, since its introduction, become the de facto standard in the literature on QKD security~\cite{BHLMO2005,Renner2005,Konig_2007,etsi_white_paper_2018,Pirandola_2020,Portmann_2022}. Over the years, however, it has also been the subject of various criticisms, which we revisited in \cref{sec_critiques}. As discussed there, they do not challenge the validity of the standard security definition but have been valuable in highlighting potential misunderstandings, which we addressed in our responses.

\section*{Acknowledgments}
This article was written as part of the Qu-Gov project, which was commissioned by the German Federal Ministry of Finance. We thank the Bundesdruckerei --- Innovation Leadership and Team for their support and encouragement. All authors acknowledge support from the ETH Zurich Quantum Center. R.W.\ acknowledges support from the Ministry of Culture and Science of North Rhine-Westphalia via the NRW-Rückkehrprogramm. 

The writing of this document greatly benefited from interactions with the wider community, which includes numerous insights, discussions, and feedback, over many years. However, any errors or omissions are entirely our own responsibility. Special thanks go to Holger Eble, Lars Kamin, Norbert L\"utkenhaus, Shlok Nahar, Ernest Y.-Z.\ Tan, and Devashish Tupkary for discussions and/or feedback on preliminary versions of this document. 

\bibliographystyle{mybibstyle}
\newcommand{\etalchar}[1]{$^{#1}$}

\section*{Appendices}
\appendix
\crefalias{section}{appendix}

\section{Proofs of statements from \cref{sec:security-definition}}

Here we provide proofs of the various statements made in \cref{sec:security-definition}. We start with \cref{lem:perr_max_existence}.

\perrmaxexistence*
\begin{proof} 
    Let $\perr: \mathcal{S}(\mathcal{H}) \times \mathcal{S}(\mathcal{H}) \mapsto \mathbb{R}_{\geq 0}$ be any function satisfying Axioms~\ref{itm:triangle_ineq}--\ref{itm:operational_bound}. Using them, we immediately find
    \begin{equation} 
        \perr(\rho, \sigma) \leq \inf \sum_{i=0}^{n-1} \perr(\omega_{i}, \omega_{i+1}) \leq \inf \sum_{i=0}^{n-1} \hat{\delta}(\omega_{i}, \omega_{i+1}) = \perr^{\max}(\rho, \sigma),
    \end{equation}
    where the infimum runs over all sets of states $\{\omega_{i}\}_{i=0}^{n}$ such that $\omega_{0} = \rho$ and $\omega_{n} = \sigma$. We will call such a set of states a \emph{path from $\rho$ to $\sigma$}.

    The above shows that $\perr^{\max}$, as constructed by~\eqref{eq:perrconstruction}, upper bounds any function $\perr$ that satisfies Axioms~\ref{itm:triangle_ineq}--\ref{itm:operational_bound}. To conclude that $\perr^{\max}$ meets \cref{def:perr_max}, it thus remains to prove that \eqref{eq:perrconstruction} also satisfies these axioms. (Axiom~\ref{itm:operational_bound} will also immediately imply that $\perr^{\max}$ is upper-bounded by~$1$, as required by its definition.) The triangle-inequality (\ref{itm:triangle_ineq}) follows because for any state $\omega$, a path from $\rho$ to $\omega$ and a path from $\omega$ to $\sigma$ can be concatenated to give a path from $\rho$ to $\sigma$ and hence $\perr^{\max}(\rho, \sigma) \leq \perr^{\max}(\rho, \omega) + \perr^{\max}(\omega, \sigma)$.
    Symmetry (\ref{itm:symmetry}) follows because any path from $\rho$ to $\sigma$ also defines a path from $\sigma$ to $\rho$ and because $\smash{\hat{\delta}}$ is symmetric. The operational bound (\ref{itm:operational_bound}) follows because $\{\rho, \sigma\}$ is a valid path and because $\hat{\delta} \leq \tilde{\delta}$. 
\end{proof}

We continue with an auxiliary lemma, which will be useful later on.
\label{sec:proofs}
\begin{appendixlemma} \label{lem: trace distance implies intuition}
    Given any two states $\rho, \sigma\in\mathcal{S}(\mathcal{H})$ satisfying $\frac{1}{2}\lVert \rho - \sigma \rVert_{1} = \varepsilon$, there exist states $\rho', \sigma'\in\mathcal{D}(\mathcal{H})$, such that
    \begin{equation}
        (1 - \varepsilon') \rho + \varepsilon' \rho' = (1 - \varepsilon') \sigma + \varepsilon' \sigma',
    \end{equation}
    where $\varepsilon' = \frac{\varepsilon}{1 + \varepsilon} \leq \varepsilon$.
\end{appendixlemma}
\begin{proof}
By the definition of the Schatten 1-norm, we have
\begin{equation}
\label{eq:def epsbar}
   \varepsilon = \frac{1}{2} \lVert \rho - \sigma \rVert_{1} = \frac{1}{2} \Tr\left(\left|\rho-\sigma\right|\right) = \frac{1}{2} \Tr\left(\left(\rho-\sigma\right)_{+}-\left(\rho-\sigma\right)_{-}\right) 
\end{equation}
where $\left(\rho-\sigma\right)_{+}$ is the positive part of the Hermitian operator $\left(\rho-\sigma\right)$ and $\left(\rho-\sigma\right)_{-}$ is its negative part. We have
\begin{equation}
\label{eq: pos and neg parts}
    \rho-\sigma = \left(\rho-\sigma\right)_{+}+\left(\rho-\sigma\right)_{-},
\end{equation} 
thus $0 = \Tr\left(\rho-\sigma\right)= \Tr\left(\left(\rho-\sigma\right)_{+}\right) + \Tr\left(\left(\rho-\sigma\right)_{-}\right)$, which, together with~\eqref{eq:def epsbar}, implies
    $\varepsilon = \Tr\left(\left(\rho-\sigma\right)_{+}\right) = -  \Tr\left(\left(\rho-\sigma\right)_{-}\right)$.

With this in mind, let $\rho'=-\frac{\left(\rho-\sigma\right)_{-}}{\varepsilon}$, and $\sigma'=\frac{\left(\rho-\sigma\right)_{+}}{\varepsilon}$ be the normalised states appearing in~\cref{eq:def epsbar}. Then
\begin{equation}
\begin{aligned}
    \rho - \sigma = (\rho - \sigma)_{+} + (\rho - \sigma)_{-} \iff& \rho - (\rho - \sigma)_{-} = \sigma + (\rho - \sigma)_{+} \\
    \iff& \rho + \varepsilon \rho' = \sigma + \varepsilon \sigma'.
\end{aligned}
\end{equation}
Dividing both sides by $1 + \varepsilon$ yields the desired statement.
\end{proof}
This suffices to show \cref{lem:perr_upper} from \cref{sec:trace-distance}.

\perrupperbound*

\begin{proof}
    Let $\varepsilon \coloneqq \frac{1}{2}\| \rho - \sigma \|_{1}$. We begin by applying the triangle-inequality (Axiom~\ref{itm:triangle_ineq}), which asserts that for any $\omega \in \mathcal{S}(\mathcal{H})$ we have that
    \begin{equation}
        \perr(\rho, \sigma) \leq \perr(\rho, \omega) + \perr(\omega, \sigma).
    \end{equation}
    By \cref{lem: trace distance implies intuition}, we can find a state $\omega$ that can be decomposed in the following two ways
    \begin{align}
        \omega =& (1 - \varepsilon') \rho + \varepsilon' \rho' \label{eq:decomp1} \\
        =& (1 - \varepsilon') \sigma + \varepsilon' \sigma' \label{eq:decomp2},
    \end{align}
    where $\varepsilon' \leq \varepsilon$, and $\rho', \sigma' \in \mathcal{S}(\mathcal{H})$. Note that by the convexity of $\mathcal{S}(\mathcal{H})$, $\omega$ is a valid quantum state.
    Using symmetry (Axiom~\ref{itm:symmetry}) as well as Axiom~\ref{itm:operational_bound} with the decomposition~\cref{eq:decomp1} yields
    \begin{equation}
        \perr(\rho, \omega) = \perr(\omega, \rho) \leq \varepsilon' \leq \varepsilon.
    \end{equation}
    Similarly, using Axiom~\ref{itm:operational_bound} with the decomposition~\eqref{eq:decomp2} gives
    \begin{equation}
        \perr(\omega, \sigma) \leq \varepsilon' \leq \varepsilon.
    \end{equation}
    Combining these statements leads to
    \begin{equation}
        \perr(\rho, \sigma) \leq 2 \varepsilon = \| \rho - \sigma \|_{1}
    \end{equation}
    as desired.
\end{proof}

To prepare for the proof of \cref{lem:perr_lower_bound} from \cref{sec:trace-distance}, we state and prove some elementary properties of the function $\tilde{\delta} : \mathcal{S}(\mathcal{H}) \times \mathcal{S}(\mathcal{H}) \mapsto \mathbb{R}_{\geq 0}$ defined within \cref{lem:perr_max_existence}.

\begin{appendixlemma} \label{lem:predist_lower_bound}
    For any two states $\rho, \sigma \in \mathcal{S}$,
    \begin{equation}
        \tilde{\delta}(\rho, \sigma) \geq \frac{1}{2} \| \rho - \sigma \|_{1}
    \end{equation}
\end{appendixlemma}
\begin{proof}
    Let $\varepsilon \coloneqq \tilde{\delta}(\rho, \sigma)$. By the definition of the infimum, for any $\varepsilon' > \varepsilon$ there exists a density operator $\sigma'$ such that
    \begin{equation}
        \rho = (1 - \varepsilon') \sigma + \varepsilon' \sigma',
    \end{equation}
    which we rewrite as
    \begin{equation}
        \rho - \sigma = \varepsilon' \sigma' - \varepsilon' \sigma.
    \end{equation}
    By the triangle inequality for the Schatten 1-norm, we then have that
    \begin{equation}
        \frac{1}{2} \| \rho - \sigma \|_{1} = \frac{1}{2} \| \varepsilon' \sigma' - \varepsilon' \sigma \|_{1} \leq \frac{1}{2} \varepsilon' \|\sigma\|_{1} + \frac{1}{2}\varepsilon' \|\sigma'\|_{1} = \varepsilon',
    \end{equation}
    where we used that $\|\sigma\|_{1} = \|\sigma'\|_{1} = 1$. Since the statement holds for any $\varepsilon' > \varepsilon$, we get the desired inequality.
\end{proof}

\begin{appendixlemma} \label{lem:predist_data_processing}
    Let $\rho, \sigma \in \mathcal{S}$ be states and let $\mathcal{E}$ be a CPTP map taking inputs from $\mathcal{S}$. Then
    \begin{equation}
        \tilde{\delta}(\rho, \sigma) \geq \tilde{\delta}(\mathcal{E}[\rho], \mathcal{E}[\sigma]).
    \end{equation}
\end{appendixlemma}
\begin{proof}
    Let $\varepsilon \coloneqq \tilde{\delta}(\rho, \sigma)$. By the definition of the infimum, for any $\varepsilon' > \varepsilon$ there exists a density operator $\sigma'$ such that
    \begin{equation}
        \rho = (1 - \varepsilon') \sigma + \varepsilon' \sigma'.
    \end{equation}
    By linearity, this immediately implies that
    \begin{equation}
        \mathcal{E}[\rho] = (1 - \varepsilon') \mathcal{E}[\sigma] + \varepsilon' \mathcal{E}[\sigma'].
    \end{equation}
    Hence $(\varepsilon', \mathcal{E}[\sigma'])$ is feasible for $\tilde{\delta}(\mathcal{E}[\rho], \mathcal{E}[\sigma])$, which in turn means that $\tilde{\delta}(\mathcal{E}[\rho], \mathcal{E}[\sigma]) \leq \varepsilon'$. Since this holds for any $\varepsilon' > \varepsilon$, we have that $\tilde{\delta}(\mathcal{E}[\rho], \mathcal{E}[\sigma]) \leq \varepsilon$ as desired.
\end{proof}

We are now ready to prove \cref{lem:perr_lower_bound} from \cref{sec:trace-distance}.

\perrlowerbound*
\begin{proof}
  Let $\hat{\delta}(\rho, \sigma)$ be defined as in \cref{lem:perr_max_existence}. \Cref{lem:predist_lower_bound} implies that $\hat{\delta}(\rho, \sigma) \geq \frac{1}{2} \| \rho-\sigma \|_1$, because the trace distance is symmetric. Using this property and the expression for $\perr^{\max}$ from \cref{lem:perr_max_existence}, we find
    \begin{equation} \label{eq:trace_dist_lower_bound}
        \perr^{\max}(\rho, \sigma) = \inf \sum_{i=0}^{n-1} \hat{\delta}(\omega_{i}, \omega_{i+1}) \geq \inf \sum_{i=0}^{n-1} \frac{1}{2} \left\| \omega_{i+1} - \omega_{i} \right\|_{1} \geq \frac{1}{2} \| \rho -  \sigma \|_{1},
    \end{equation}
    where the second inequality is the triangle inequality of the trace distance. 

    It thus remains to prove Properties~\ref{itm:positive_definite} and~\ref{itm:data_processing}. Positive-definiteness (\ref{itm:positive_definite}) follows immediately from \cref{eq:trace_dist_lower_bound} since $\perr^{\max}(\rho, \sigma) = 0$ implies $\|\rho - \sigma\|_1 = 0$ and hence $\rho = \sigma$. It is also clear that $\perr^{\max}(\rho, \rho) = 0$.
    Finally, the data processing inequality (\ref{itm:data_processing}) holds because $\{\mathcal{E}[\omega_{i}]\}_{i}$ is a valid path from $\mathcal{E}[\rho]$ to $\mathcal{E}[\sigma]$ (as defined within the proof of \cref{lem:perr_max_existence}) and because $\smash{\hat{\delta}}$ satisfies the data-processing inequality, which can be seen from \cref{lem:predist_data_processing}.
\end{proof}

Finally, we prove a lemma that immediately implies the statement made in the remark on $\varepsilon$-correctness just below \cref{def: trace distance criterion}.

\begin{appendixlemma} \label{lem:correctness_alt_form}
    Let $\rho_{K_A K_B}^{\mathrm{real}}$ be the state of Alice and Bob's keys. Furthermore, define $\hat{\rho}_{K_A K_B}$ to be the same as $\rho_{K_A K_B}^\mathrm{real}$ but with Bob's key replaced by Alice's key. More precisely, for
    \begin{equation}
        \rho_{K_A K_B}^{\mathrm{real}} = \sum_{k_A, k_B} P_{K_A K_B}(k_A, k_B) \ketbra{k_A}_{K_A} \otimes \ketbra{k_B}_{K_B}
    \end{equation}
    define
    \begin{equation}
    \begin{aligned}
        \hat{\rho}_{K_A K_B} =& \sum_{k_A,k_B} P_{K_A K_B}(k_A, k_B) \ketbra{k_A}_{K_A} \otimes \ketbra{k_A}_{K_B} \\
        =& \sum_{k_A, k_B} P_{K_A}(k_A) \delta_{k_A, k_B} \ketbra{k_A}_{K_A} \otimes \ketbra{k_B}_{K_B}.
    \end{aligned}
    \end{equation}
    Then
    \begin{equation}
        \frac{1}{2} \norm{\rho^{\mathrm{real}}_{K_A K_B} - \hat{\rho}_{K_A K_B}}_1 = \mathrm{Pr}[K_A \neq K_B].
    \end{equation}
\end{appendixlemma}
\begin{proof}
    We calculate
    \begin{equation}
    \begin{aligned}
        &\frac{1}{2} \norm{\rho^{\mathrm{real}}_{K_A K_B} - \hat{\rho}_{K_A K_B}}_1 \\
        =& \frac{1}{2} \sum_{k_A, k_B} \abs{P_{K_A K_B}(k_A, k_B) - P_{K_A}(k_A) \delta_{k_A, k_B}} \\
        =& \frac{1}{2} \sum_{k_A \neq k_B} P_{K_A K_B}(k_A, k_B) + \frac{1}{2} \sum_{k_A = k_B} \abs{P_{K_A K_B}(k_A, k_B) - P_{K_A}(k_A) \delta_{k_A, k_B}}.
    \end{aligned}
    \end{equation}
    For the second term, we compute
    \begin{equation}
    \begin{aligned}
        \sum_{k_A = k_B} \abs{P_{K_A K_B}(k_A, k_B) - P_{K_A}(k_A) \delta_{k_A, k_B}}
        =& \sum_{k_A = k_B} \abs{P_{K_A K_B}(k_A, k_B) - P_{K_A}(k_A)} \\
        =& \sum_{k_A = k_B} P_{K_A}(k_A) - P_{K_A K_B}(k_A, k_B) \\
        =& 1 - \mathrm{Pr}[K_A = K_B] \\
        =& \mathrm{Pr}[K_A \neq K_B],
    \end{aligned}
    \end{equation}
    where, for the second equality, we used that $P_{K_A}(k_A) \geq P_{K_A K_B}(k_A, k_A)$. Putting everything together, we have
    \begin{equation}
        \frac{1}{2} \norm{\rho^{\mathrm{real}}_{K_A K_B} - \hat{\rho}_{K_A K_B}}_1 = \frac{1}{2} \mathrm{Pr}[K_A \neq K_B] + \frac{1}{2} \mathrm{Pr}[K_A \neq K_B] = \mathrm{Pr}[K_A \neq K_B],
    \end{equation}
    as claimed.
\end{proof}

\end{document}